\documentclass[a4paper,onecolumn,unpublished]{quantumarticle}
\pdfoutput=1

\usepackage{graphicx} 
\usepackage{subfiles}
\usepackage{pdfpages}

\usepackage{marginfix}
\usepackage{multicol}

\usepackage[utf8]{inputenc} 
\usepackage[T1]{fontenc}    
\usepackage{hyperref}       
\hypersetup{
        colorlinks = true,
        linkcolor = blue!70!black,
        anchorcolor = blue!70!black,
        citecolor = red!70!black,
        filecolor = blue!70!black,
        urlcolor = blue!70!black
}
\usepackage{url}            
\usepackage{amsfonts}       
\usepackage{nicefrac}       
\usepackage{xcolor}         
\usepackage{stmaryrd}       
\usepackage{multirow}       
\usepackage{float}
\usepackage{subcaption}
\usepackage[font=small]{caption}
\usepackage{mathtools}
\usepackage{amssymb}
\usepackage{amsthm}
\usepackage{thmtools}
\usepackage{braket}
\def\<{\langle}
\def\>{\rangle}
\usepackage{csquotes}
\usepackage{xspace}
\usepackage{outlines}

\usepackage{tikz-cd}
\usepackage{tikzit}
\usepackage{mathrsfs}
\usepackage{pgfplots}
\usepackage[shortlabels]{enumitem}
\usepackage[style=numeric, maxbibnames=99, sortcites]{biblatex}
\usepackage{todonotes}

\tikzstyle{Z dot}=[inner sep=0mm, minimum size=2mm, shape=circle, draw=black, fill=zxGreen, tikzit fill={rgb,255: red,216; green,248; blue,216}, outer sep=-0.5mm]
\tikzstyle{Z phase dot}=[draw=black, fill=zxGreen, shape=rectangle, minimum size=4.5mm, rounded corners=1.8mm, inner sep=0.5mm, outer sep=-0.5mm, scale=0.8, tikzit shape=circle, tikzit fill={rgb,255: red,216; green,248; blue,216}, font={\footnotesize\boldmath}]
\tikzstyle{X dot}=[shape=circle, draw=black, fill=zxRed, tikzit fill={rgb,255: red,221; green,165; blue,165}, inner sep=0 mm, minimum size=2 mm, outer sep=-0.5mm]
\tikzstyle{X phase dot}=[Z phase dot, draw=black, fill=zxRed, tikzit fill={rgb,255: red,221; green,165; blue,165}]
\tikzstyle{H box}=[fill=zxHad, draw=black, shape=rectangle, inner sep=0.6mm, minimum height=1.5mm, minimum width=1.5mm, tikzit fill=yellow, font={\footnotesize\boldmath}]
\tikzstyle{box}=[draw=black, shape=rectangle, fill=white, minimum size=1em, inner sep=0.2em, scale=0.85, font={\scriptsize}, outer sep=-0.5mm]
\tikzstyle{black dot}=[fill=black, draw=black, shape=circle, inner sep=1pt]
\tikzstyle{sLabel}=[font={\scriptsize}, tikzit draw=black, auto]
\tikzstyle{not}=[draw=black, circle, addcross, minimum size=2mm, outer sep=-0.5mm, inner sep=0mm]
\tikzstyle{Z dot thick}=[inner sep=0mm, minimum size=2mm, shape=circle, draw=black, fill=zxGreen, tikzit fill={rgb,255: red,216; green,248; blue,216}, outer sep=-0.5mm, line width=1pt]
\tikzstyle{X dot thick}=[inner sep=0mm, minimum size=2mm, shape=circle, draw=black, fill=zxRed, tikzit fill={rgb,255: red,221; green,165; blue,165}, outer sep=-0.5mm, line width=1pt]
\tikzstyle{fault-location}=[fill=white, draw=black, shape=circle, minimum size=2mm, inner sep=0mm, outer sep=-0.5 mm, regular polygon, regular polygon sides=8, font={\tiny}]
\tikzstyle{fault-free-location}=[fill=white, draw={rgb,255: red,177; green,98; blue,255}, shape=circle, minimum size=2mm, inner sep=0mm, outer sep=-0.5 mm, regular polygon, regular polygon sides=8, font={\tiny}]
\tikzstyle{fault-location-faulty}=[fill=white, draw={rgb,255: red,191; green,0; blue,64}, shape=circle, minimum size=2mm, inner sep=0mm, outer sep=-0.5 mm, regular polygon, regular polygon sides=8, minimum size=3mm, font={\tiny}]
\tikzstyle{new style 0}=[fill=white, draw=black, shape=circle]

\tikzstyle{dashed-line}=[-, style=dashed, draw={rgb,255: red,128; green,128; blue,128}]
\tikzstyle{hadamard}=[-, style=dashed, draw=blue]
\tikzstyle{X Web}=[-, preaction={line width=1mm, draw=zxDarkRed}, tikzit draw=red]
\tikzstyle{Z Web}=[-, preaction={line width=1.7mm, draw=zxDarkGreen}, tikzit draw=green]
\tikzstyle{XZ Web}=[-, preaction={line width=1.8mm, draw=zxDarkGreen}, preaction={line width=1mm, draw=zxDarkRed}, tikzit draw=blue]
\tikzstyle{braceedge}=[-, decorate, decoration={brace, amplitude=2mm, raise=-1mm}]
\tikzstyle{arrow}=[->]
\tikzstyle{fault-free}=[-, draw={rgb,255: red,177; green,98; blue,255}, line width=1pt]
\tikzstyle{new edge style 0}=[-, fill={rgb,255: red,216; green,216; blue,216}, draw=black]
\tikzstyle{new edge style 1}=[->, draw={rgb,255: red,177; green,98; blue,255}, line width=1pt]
\tikzstyle{fault-accounting-component}=[-, style=dashed, draw={rgb,255: red,0; green,128; blue,128}]
\tikzstyle{new edge style 2}=[-, style=dashed, draw={rgb,255: red,0; green,128; blue,128}]
\tikzstyle{thick}=[-, line width=1pt]

\input{floquet.tikzdefs}

\theoremstyle{definition}
\declaretheorem[numberwithin=section]{theorem}
\declaretheorem[sibling=theorem]{definition}
\declaretheorem[sibling=theorem]{proposition}
\declaretheorem[sibling=theorem]{remark}
\declaretheorem[sibling=theorem]{corollary}

\def\bR{\begin{color}{red}}
\def\bB{\begin{color}{blue}}
\def\bG{\begin{color}{green}}
\def\bP{\begin{color}{purple}}
\def\bC{\begin{color}{cyan}}
\def\bM{\begin{color}{magenta}}
\def\bL{\begin{color}{lime}}
\def\e{\end{color}}

\usepackage[normalem]{ulem}

\newcommand{\interp}[1]{#1}
\newcommand{\code}[1]{\left\llbracket#1\right\rrbracket}

\definecolor{cbpink}{RGB}{214,130,211}
\definecolor{cbyellow}{RGB}{241,131,108}
\definecolor{cbblue}{RGB}{110,148,189}

\definecolor{cbred}{RGB}{162,4,162}
\definecolor{cborange}{RGB}{251,145,10}
\definecolor{cbgreen}{RGB}{5,162,162}
\definecolor{cbcyan}{RGB}{204,234,207}

\newcommand{\TextOCM}{\textsc{OCM}\xspace}
\newcommand{\TextFusion}{\textsc{Fusion}\xspace}

\newcommand{\TextCopy}{\textsc{Copy}\xspace}

\newcommand{\TextPiCommute}{\textsc{Pi-Copy}\xspace}

\newcommand{\pPhis}{p_{\mathrm{phys}}}

\newcommand{\pMem}{p_{\mathrm{mem}}}
\newcommand{\pTwo}{p_{\mathrm{2}}}

\newcommand{\pSPAM}{p_{\mathrm{SPAM}}}

\title{Fault Tolerance by Construction}
\author{Benjamin Rodatz, Boldizsár Poór, Aleks Kissinger}
\affiliation{University of Oxford, Oxford, UK}

\begin{document}

\allowdisplaybreaks
\maketitle

 \begin{abstract}
   A key challenge in fault-tolerant quantum computing is synthesising and optimising circuits in a noisy environment, as traditional techniques often fail to account for the effect of noise on circuits. 
In this work, we propose and numerically verify a framework for designing fault-tolerant quantum circuits that are \emph{correct by construction}.
The framework starts with idealised specifications of fault-tolerant gadgets and refines them using provably sound basic transformations.

To reason about manipulating circuits while preserving their error correction properties, we define \emph{fault equivalence}; two circuits are considered fault-equivalent if all undetectable faults on one circuit have a corresponding fault on the other.
This guarantees that the effect of undetectable faults on both circuits is the same. 
We argue that fault equivalence is a concept that is already implicitly present in the literature. 
Many problems, such as state preparation and syndrome extraction, can be naturally expressed as finding an implementable circuit that is fault-equivalent to an idealized specification. 

To utilize fault equivalence in a computationally tractable manner, we adapt the ZX calculus, a diagrammatic language for quantum computing.
We restrict its rewrite system to not only preserve the underlying linear map but also fault equivalence, i.e.\@ the circuit's behaviour under noise.
Enabled by our framework, we verify, optimise, and synthesise new and efficient circuits for syndrome extraction and cat state preparation.
We numerically confirm the improved performance of our optimised circuits in simulation.
We anticipate that fault equivalence can capture and unify different approaches in fault-tolerant quantum computing, paving the way for an end-to-end circuit compilation framework.
 \end{abstract}

\tableofcontents

\section{Introduction}

A key challenge in building scalable and reliable quantum computers is the suppression of noise caused by interaction with the environment and imperfect gate operations~\parencite{shorFaulttolerantQuantum1996}.
Fault-tolerant quantum computation (FTQC) offers a collection of techniques to deal with these issues.
FTQC usually employs a quantum error-correcting code to introduce redundancy in quantum data to allow errors to be detected and corrected, then describes how one can implement computations on the encoded data and rounds of error correction in such a way that faults from imperfect gate realizations remain localized.

It is often useful to break the task of compiling a large, fault-tolerant computation into building blocks, or `gadgets', which give physical quantum circuit implementations for the basic state preparations, gates, and measurements needed by a computation.
Constructing fault-tolerant versions of these building blocks is a crucial but difficult task.
A key challenge is that faults in multi-qubit operations can create correlated errors in quantum data, and these errors can propagate through the circuit in non-trivial ways.
Verifying that a gadget is robust against all such possibilities is a difficult problem, often requiring resource-intensive simulations or complex analytical proofs~\parencite{gotoMinimizingResource2016,chenVerifyingFaultTolerance2025}.

While many quantum circuit synthesis techniques exist to decompose complex quantum state preparations and unitaries into basic gates, most techniques are insufficient for the construction of fault-tolerant gadgets~\parencite{hwangFaulttolerantCircuit2022}. This is because they focus only on implementing the ideal linear map and do not account for the ways in which faults might alter this map.
However, different circuits that realize the same unitary can behave very differently under the influence of noise, depending on how they propagate faults.
A \textit{de facto} solution to this problem is to build a candidate circuit implementing a fault-tolerant computation, then check its behaviour under noise either analytically or using a classical simulation tool like Stim~\parencite{gidneyStimFast2021}.

Building fault-tolerant gadgets in this way requires a great deal of skill, trial and error, computational resources, or all of the above.
In this paper, we propose a new methodology for compiling fault-tolerant computations that are \textit{correct by construction}.
That is, we start with an idealised computation that has good properties with respect to fault tolerance, but is not implementable, i.e.\@ it is not expressed in terms of the basic, potentially noisy operations available on a quantum device.
We then transform that computation step-wise into something that is expressed in terms of these basic operations in such a way that the behaviour of the computation under faults is preserved.
This methodology of starting with an idealised \textit{specification} and \textit{refining} it into an implementable program using a sequence of provably sound basic transformations has been a cornerstone of formal software verification since at least the 1980s~\parencite{bjornerFormalSpecification1982,backRefinementCalculus1998}.
As we show, it can be fruitfully applied to the problem of fault-tolerant quantum compilation as well.

To adapt these ideas to FTQC, we introduce and formalize the concept of \emph{fault equivalence}:
Two circuits are considered fault-equivalent if all undetectable faults on one circuit have a corresponding fault on the other circuit of equal or lesser weight.
Practically, this notion ensures that circuit transformations do not introduce new, undetectable faults that propagate badly.
This concept unifies and generalises many related conditions and criteria in the literature, and gives us a way to reason about both the construction of specific fault-tolerant gadgets, such as state preparations and syndrome extraction circuits, as well as the fault-tolerant computation as a whole.

After establishing the concept and basic properties of fault equivalence for circuits in the first part of the paper, we will demonstrate how the ZX calculus, a useful tool for circuit rewriting and optimisation, can be adapted to take faults into account as well.
The ZX calculus represents quantum computations using certain undirected graphs called ZX diagrams, which can be transformed using a handful of simple rewrite rules~\parencite{coeckeInteractingQuantumObservables2008}.

It has been applied extensively to quantum circuit synthesis and optimization~\parencite{duncanGraphtheoreticSimplification2020,cowtanPhaseGadget2020,Staudacher2022reducing,winderlRecursivelyPartitioned2023,wille2022basis,gogioso2023annealing,riuReinforcementLearning2025,staudacher2024multicontrolled}, and has recently been of some interest in representing and reasoning about quantum error correction and FTQC~\parencite{kissingerPhasefreeZXDiagrams2022,kissingerScalableSpiderNests2024, huangGraphicalCSSCode2023, khesinUniversalGraph2025, debeaudrapZXCalculus2020, cowtanCSSCode2024,townsend-teagueFloquetifyingColourCode2023,bombinUnifyingFlavorsFault2024,rodatzFloquetifyingStabiliser2024,mcewenRelaxingHardwareRequirements2023,gidneyPairMeasurement2023,garvie2018Verifying,duncan2013verifying,}.
An important property of ZX diagrams is that they are time-agnostic, so a single diagram can have multiple interpretations and/or implementations depending on how one chooses a time orientation and which sets of nodes correspond to basic gates and measurements.
This lends the flexibility to completely re-interpret a quantum computation during the course of ZX rewriting, which has no analogue in traditional circuit synthesis techniques.

Adopting the ZX calculus for our purposes requires that equivalence under noise is preserved both (1) when mapping between quantum circuits and ZX diagrams, and (2) when performing ZX rewrites.
To achieve the former, we introduce the concept of a fault-equivalent representation.
This property ensures that the behaviour of the quantum circuit under some noise model is fully accounted for by the behaviour of the ZX diagram under its noise model. 
To achieve the latter, we define fault-equivalent rewrites, a subset of ZX rewrite rules that are proven to preserve fault equivalence in ZX diagrams.
Standard ZX rules only preserve the ideal (i.e.\@ fault-free) linear map represented by a ZX diagram, while these fault-equivalent rewrites preserve a more fine-grained notion of equivalence that also considers possible faults.

Building on these, we provide a constructive framework for fault-equivalent circuit synthesis:
\begin{enumerate}
  \item A specification, such as an idealised circuit under an idealised noise model, is translated into a fault-equivalent ZX diagram.
  \item The diagram is transformed and optimized using only fault-equivalent rewrites of ZX diagrams, which guarantee that the fault-tolerant properties of the idealised specification are preserved.
  \item A new, optimized circuit is extracted from the rewritten diagram that is guaranteed to be fault-equivalent to the idealised specification.
\end{enumerate}
Notably, the framework guarantees the fault tolerance of the derived circuits despite their lack of obvious structural indicators of fault tolerance, such as transversal gates~\parencite{shorFaulttolerantQuantum1996,steaneActiveStabilization1997}.

The key conceptual shift behind this new paradigm for fault-tolerant circuit synthesis is to analyze the relationship of quantum circuits under noise, rather than treating computations in isolation, as formalized by fault equivalence. 
More broadly, this relational viewpoint reveals that many apparently distinct aspects of fault-tolerant quantum computing are in fact expressions of a single unifying principle. 
Within this framework, fault-tolerant circuit verification amounts to determining whether a quantum circuit can be transformed back to some given specification, while fault-tolerant circuit optimization involves transforming a circuit into a more efficient equivalent one. 
Framing these problems in terms of circuit transformations enables productive exchange between synthesis, verification, and optimization techniques.

We demonstrate the effectiveness of this framework by applying it to the verification, optimisation, and design of elementary FTQC gadgets, leading to the discovery of new efficient circuits with guaranteed fault tolerance. 
We show how known methods, namely Shor- and Steane-style syndrome extraction, can be derived straightforwardly in this framework.
Additionally, we show that fault-equivalent rewrites directly yield a novel cat-state preparation method which is applicable to codes with any distance.
Finally, we use our framework to construct new provably sound optimisations of Shor and Steane-style syndrome extraction circuits, which use fewer ancillae and gates (\autoref{fig:steane-comparison}).
For the Steane-style syndrome extraction, we use STIM simulations~\parencite{gidneyStimFast2021} to compare the performance of the original and optimised circuits. 
We show that, as predicted, the optimised circuits outperform the original circuits.

The paper is organized into two parts.
In the first part, building on foundational spacetime code perspective of noise in quantum circuits \parencite{delfosseSpacetimeCodes2023,gottesmanOpportunitiesChallenges2022,baconSparseQuantum2017}, we define the central concept of fault equivalence, and demonstrate its utility by re-expressing established concepts within this new framework.
In the second part, we adapt the ZX calculus to reason about fault equivalence, define faults on ZX diagrams, and establish fault-equivalent rewrite rules.
We then use this framework to verify, synthesize, and optimize novel circuits for cat state preparation and syndrome extraction.

\begin{figure}[htb]
  \begin{subfigure}[t]{.5\textwidth}
    \[
 \input{./figures/steane-7-initial-circuit.tikz}
    \]
    \caption{Standard Steane-style syndrome extraction.}
    \label{fig:steane-initial}
  \end{subfigure}
  \begin{subfigure}[t]{.5\textwidth}
    \[
 \input{./figures/steane-7-final-circuit.tikz}
    \]
    \caption{Optimized circuit derived via fault-equivalent rewrites.}
    \label{fig:steane-final}
  \end{subfigure}

  \caption{Comparison of standard and optimized Steane-style syndrome extraction circuits for the $\code{7,1,3}$ Steane code using fault-tolerant state preparation of \textcite{gotoMinimizingResource2016}.}
  \label{fig:steane-comparison}
\end{figure}

\subsection{Related Work}
The construction of efficient and reliable fault-tolerant circuits is a central challenge in quantum computing, and a variety of approaches have been developed to address it.
These methodologies can be broadly categorized into foundational formalisms, automated heuristic search, protocol-specific design, and formal verification.
Our work introduces a novel framework that bridges formal methods with practical circuit synthesis using the concept of fault equivalence.

The spacetime perspective on faults, which is central to our work, builds upon several foundational ideas.
While the idea of propagating Pauli errors through Clifford circuits has been fundamental to fault-tolerant computation since its inception, the formal treatment of spacetime locations in the style we use in this paper was introduced by \textcite{baconSparseQuantum2017}.
In their work, the \emph{spackle map} was introduced to model how a single Pauli fault propagates and branches through gates, effectively creating a directed web of correlated errors across spacetime.
This perspective was elevated into a circuit-centric paradigm by \textcite{gottesmanOpportunitiesChallenges2022}, who argued for a fundamental shift in focus.
Instead of concentrating on a static error-correcting code that a circuit acts upon, this view treats the circuit's dynamics itself as defining the code, proposing that the primary goal of a fault-tolerant scheme should be to identify the spacetime location of a fault, rather than just its effect on a static codespace.
The modern mathematical language for this paradigm, which we make use of extensively in Part I of this paper, is provided by the theory of spacetime codes from \textcite{delfosseSpacetimeCodes2023}.

The primary tool we adapt for this purpose is the ZX calculus.
While it has been employed by \textcite{bombinUnifyingFlavorsFault2024} as a descriptive language to unify different models of topological fault tolerance, our work focuses on making it a constructive tool for fault-tolerant circuit synthesis.
This approach is a direct continuation of our prior work in \textcite{rodatzFloquetifyingStabiliser2024}, where the core concept of fault-equivalent rewrites was first introduced under the name \enquote{distance-preserving rewrites}.
The present work expands significantly on that foundation, introducing a more general framework and a more rigorous mathematical treatment of its core principles.

Pauli webs were first formulated by \textcite{bombinUnifyingFlavorsFault2024} as a graphical notation to identify and analyze checks and stabilizers of ZX diagrams.
A direct analogue, termed a \emph{detecting region}, was independently developed in the circuit formalism by \textcite{mcewenRelaxingHardwareRequirements2023}, where it was used to obtain insights and develop new, more hardware-efficient circuits for implementing QEC codes.
Additionally, Pauli webs were employed by \textcite{townsend-teagueFloquetifyingColourCode2023} to relate the detecting regions of a static quantum error correction code with those of a dynamic Floquet code derived via ZX transformations.
\textcite{litinskiBlockletConcatenation2025} utilized Pauli webs to keep track of and understand information flow of fault-tolerant protocols constructed with concatenation and transversal gates.
Furthermore, \textcite{wanPauliWebs2025} used Pauli webs to graphically verify the correctness of the transversal $\ket{Y} = S\ket{+}$ state initialization scheme for the surface code.

With this context, we can contrast our framework against other techniques.
A dominant approach for circuit synthesis is automated discovery using heuristic search and optimization.
Techniques based on Reinforcement Learning (RL)~\parencite{zenQuantumCircuit2024}, sometimes combined with other methods like unitary diagonalization~\parencite{weidenLearningSynthesize2024}, can be used to discover compact, hardware-efficient circuits.
Similarly, Satisfiability (SAT) and Satisfiability Modulo Theories (SMT) solvers can find gate- or depth-optimal circuits by translating the design problem into a set of logical constraints~\parencite{pehamAutomatedSynthesis2025,shuttyDecodingMerged2022}.
These search-based methods have proven effective in creating novel circuits, including those that leverage sophisticated techniques like flag qubits to detect errors~\parencite{chamberlandFlagFaulttolerant2018,chaoFlagFaultTolerant2020}.
However, in these approaches, fault tolerance is typically a property to be optimized within a complex reward function or verified after the fact.
In contrast, our framework replaces heuristic search with formal derivation, ensuring that fault tolerance is provably preserved through every transformation.

At the other end of the spectrum lies formal verification, which aims to prove that a completed circuit design is indeed fault-tolerant.
Recent advances, for example, use quantum symbolic execution to build automated verifiers for this purpose~\parencite{chenVerifyingFaultTolerance2025}.
While verification provides assurance for a final design, it relies on circuits that have already been constructed.
Our approach is to build fault-tolerant circuits that are correct by construction, obviating the need for a separate verification step.

Independent work introduced an alternative framework for fault-tolerantly transforming quantum circuits~\parencite{pesah2025faulttoleranttransformationsspacetimecodes}, similarly leveraging the spacetime perspective of quantum circuits~\parencite{baconSparseQuantum2017,gottesmanOpportunitiesChallenges2022,delfosseSpacetimeCodes2023}.
Instead of directly transforming the circuits under noise, they translate them into chain complexes and view rewrites as chain maps.

\part{Faults on Circuits}

\section{Faults on Circuits}\label{sec:faults-on-circs}

\subsection{Clifford circuits and spacetime locations}\label{sec:spacetime-locs}

Following \textcite{baconSparseQuantum2017,gottesmanOpportunitiesChallenges2022, delfosseSpacetimeCodes2023}, we consider Pauli faults in spacetime, defined as Pauli actions occurring at time steps between gates.

We use $\mathcal P^n$ and $\overline{\mathcal{P}^n}$ to respectively denote $n$-qubit Pauli operators and the quotient of $\mathcal P^n$ by its centre $\{\pm I, \pm iI \}$.


We view gates in a circuit as mapping the state at time step $t$ to a new state at time step $t + 1$.
In the next section, we will define faults as certain Pauli operators acting on the qubits at each time step, i.e.\@ in between the gates. To do this rigorously, we define the \textit{spacetime locations} in a circuit, as well as the linear operator obtained by \textit{applying} a Pauli operator at those locations.

\begin{definition}[Spacetime locations]
 Let $C$ be a quantum circuit.
 The \emph{spacetime locations} $\mathcal L$ of $C$ consist of pairs $(t,q)$ for each timestep $t$ and qubit $q$ that exists at timestep $t$.
 A \emph{fault} $F$ on $C$ is defined as an element in $\overline{\mathcal{P}^{|\mathcal L|}}$, indicating the action of $F$ on each spacetime location of $C$.
\end{definition}

See, for example, \autoref{fig:fault-examples} (a) where the octagons indicate the potential locations of Pauli faults\footnote{Prior work~\parencite{baconSparseQuantum2017,gottesmanOpportunitiesChallenges2022,delfosseSpacetimeCodes2023} used empty circles to indicate fault-locations. To avoid confusion with the spiders of the ZX calculus, we instead use octagons.}.

\begin{figure}[htb]
    \centering
    \begin{minipage}{0.45\textwidth}
        \centering
        \scalebox{1.15}{\input{./figures/02-errors/example-circuit-fault-locations.tikz}}
        \caption*{(a)}
    \end{minipage}
    \begin{minipage}{0.45\textwidth}
        \centering
        \scalebox{1.15}{\input{./figures/02-errors/example-circuit-measurement-fault.tikz}}
        \caption*{(b)}
    \end{minipage}
    \caption{A circuit on five qubits over four time-steps. We have (a) The spacetime locations of the circuit shown as octagons. (b) Applying a Pauli $X_jX_k \in \overline{\mathcal P}^{|\mathcal L|}$ to a circuit $C$. In this case, the resulting linear map contains a projection onto the $-1$ eigenspace of $Z \otimes Z$ and projections onto the $+1$ eigenspaces of the other two Pauli measurements.
    \label{fig:spacetime-locations}}
\end{figure}

For the purposes of this paper, we will focus on Clifford circuits, which consist of Clifford unitary gates, preparations of single-qubit states $\{ |0\>, |1\>, |{+}\>, |{-}\> \}$, and measurements of arbitrary Pauli operators.

Notably, circuits with spacetime locations define not just a single linear map, but a whole family of them, indexed by Pauli operators $P \in \overline{\mathcal{P}^{|\mathcal L|}}$.

\begin{definition}[Applying Paulis to a circuit]
    \label{def:pauli-application}
For a circuit $C$ with spacetime locations $\mathcal L$ and a Pauli operator $P \in \overline{\mathcal{P}^{|\mathcal L|}}$, let $C^P$ be the linear map obtained by placing a Pauli $P_j$ at location $j \in \mathcal L$, where $P \propto P_1 \otimes \ldots \otimes P_{|\mathcal L|}$. We call $C^P$ the \textit{application} of $P$ to $C$. Two Paulis are considered equivalent with respect to $C$, written $P \sim Q$ if $C^P = C^Q$, and a Pauli $P$ is called \textit{trivial} with respect to $C$ if $P \sim I$.
\end{definition}

Note that if $C$ contains Pauli measurements, we treat them as projectors onto the $+1$ eigenspace when we compute the linear map as in \autoref{def:pauli-application}. This corresponds to obtaining the $+1$ outcome for that measurement. If we wish to compute the map corresponding to the other outcome, we can introduce an anti-commuting Pauli operator before and after the measurement, as in \autoref{fig:spacetime-locations}(b). Hence, all the branches of a computation represented by a circuit $C$ involving Clifford gates and Pauli measurements can be computed as applications of Paulis to $C$.

\subsection{Faults}

Our main purpose in introducing spacetime locations is to represent the effect of faults on a circuit. Faults are simply Paulis drawn from a fixed subgroup $\langle \mathcal F\rangle \subseteq \overline{\mathcal{P}^{|\mathcal L|}}$. This group is generated by a chosen set $\mathcal F$ of \textit{atomic faults} we wish to consider. 
Atomic faults make up the faults that can occur, if a single fault event happens, e.g.\@ if a single measurement flip or a single gate error occurs. 
These correspond to the independent error mechanisms from STIM \parencite{gidneyStimFast2021}. 
Fixing this set of atomic faults is what defines our \textit{noise model}:

\begin{definition}[Adversarial noise model, fault weight]
 Given a quantum circuit $C$ with spacetime locations $\mathcal L$, an \emph{adversarial noise model} $\mathcal F$ is a set of atomic faults in $\overline{\mathcal{P}^{|\mathcal L|}}$, each representing the effect of a single fault event.
 The set of potential faults under $\mathcal F$ is the group $\langle \mathcal F\rangle \subseteq \overline{\mathcal{P}^{|\mathcal L|}}$.
 The \emph{weight} of a fault $F \in \langle \mathcal F\rangle$, denoted by $wt(F)$, is the minimal number of atomic faults in $\mathcal F$ required to generate $F$.
\end{definition}

Note that throughout this work, we only reason about the worst-case behaviour of circuits under noise, i.e.\@ how the circuits behave if at most $n$ faults occur.
This is sufficient for reasoning about circuit synthesis and optimisation of fault-tolerant gadgets.
As such, our noise models do not assign probabilities to faults.
We leave extending this framework to average-case behaviour under stochastic noise as future work. 

\begin{figure}[htb]
    \centering
    \begin{minipage}{0.45\textwidth}
        \centering
        \scalebox{1.15}{\input{./figures/02-errors/example-circuit-measurement-fault.tikz}}
        \caption*{(a)}
    \end{minipage}
    \begin{minipage}{0.45\textwidth}
        \centering
        \scalebox{1.15}{\input{./figures/02-errors/example-circuit-cnot-fault.tikz}}
        \caption*{(b)}
    \end{minipage}

    \caption{Examples of atomic faults on a circuit. (a) A measurement flip of the Pauli-$ZZ$ measurement created by placing a Pauli $X$ before and after the measurement. (b) Two Paulis after a CNOT, which count as a fault of weight one.}
    \label{fig:fault-examples}
\end{figure}

While there are a variety of potential noise models to consider, we will consider the following: 
\begin{definition}[Circuit-level noise]
 The \emph{circuit-level noise model} for some quantum circuit $C$ consists of the following atomic faults:
    \begin{itemize}
        \item For each wire: 
        \begin{itemize}
            \item \textbf{Qubit flips:} A fault acting only on that wire.  
        \end{itemize}      
        \item For each gate (including Pauli measurements):
        \begin{itemize}
            \item \textbf{Gate faults:} Faults acting on any subset of the output wires of the gate.
        \end{itemize}
        \item For each Pauli measurement:
        \begin{itemize}
            \item \textbf{Measurement flip:} An anticommuting Pauli applied immediately before and after the measurement.
            \item \textbf{Measurement flip + gate faults:} A measurement flip and a fault acting on any subset of the outputs of the measurement.
        \end{itemize}
    \end{itemize}
\end{definition}

\autoref{fig:fault-examples}(a) shows how measurement flips can be represented using this framework.
Despite affecting multiple spacetime locations, in this noise model, the measurement flip is an atomic fault, counting as one fault event.
\autoref{fig:fault-examples}(b) shows a multi-qubit fault created by a single fault event on the CNOT. 

As a mathematical reasoning tool, we additionally consider subsets of circuit-level noise, where we idealise certain operations and wires as fault-free:
\begin{definition}[Idealised submodel of circuit-level noise]
 Let $C$ be a quantum circuit where we idealise some of the wires and gates as fault-free.
 Multi-qubit Pauli measurements can be idealised as either fault-tolerant or fault-free. 
 The idealised submodel of circuit-level noise then consists of \textbf{all the faults of circuit-level noise except for the following:}
    \begin{itemize}
        \item For each \emph{fault-free} wire: 
        \begin{itemize}
            \item Removing all \textbf{qubit flips} affecting that wire. 
        \end{itemize}       
        \item For each \emph{fault-free} multi-qubit gate:
        \begin{itemize}
            \item Removing all \textbf{gate faults} on its outputs.
        \end{itemize}
        \item For each \emph{fault-tolerant but not fault-free} multi-qubit Pauli measurement:
        \begin{itemize} 
            \item Removing all \textbf{measurement flip + multi-qubit gate faults} only \emph{leaving measurement flips and measurement flips + single-qubit gate faults}.
        \end{itemize}
        \item For each \emph{fault-free} multi-qubit Pauli measurement:
        \begin{itemize} 
            \item Removing all \textbf{measurement flip and gate faults}.
        \end{itemize}
    \end{itemize}
\end{definition}

\begin{figure}
    \centering
    \input{./figures/02-errors/example-circuit-fault-free.tikz}

    \caption{Bold, purple wires and operations indicate components idealized as fault-free under the chosen noise model. The $ZZ$ measurement is idealised as fault-tolerant, and the $XZ$ measurement is idealised as completely fault-free (as indicated by the colour of the measurement arrow at the top).}
    \label{fig:idealised-wires}
\end{figure}

As such, idealised submodels of circuit-level noise consist of strictly fewer atomic faults, where, by assumption, we impose some of the circuit-level noise faults to be impossible.
While this is not realistic for circuits that are run on hardware, it serves as an essential mathematical reasoning tool throughout this work.

For multi-qubit Pauli measurements, we introduce two levels of idealisation: fault-tolerant and the stronger notion of fault-free.
While the two idealisations might not seem immediately obvious at first, they will become clearer throughout the rest of this work.
In short, oftentimes fault-tolerant implementations of Pauli measurements are assumed to be free of multi-qubit faults on the outputs but not free of all faults, in that faults that occur during the measurements can still create a measurement flip and/or faults of at most the same weight on the outputs.

We indicate the noise model under consideration by drawing the idealised edges and gates in purple.
We indicate fault-tolerant Pauli measurements by drawing the box in purple but leaving the measurement arrow at the top in black, and we draw fault-free Pauli measurements completely in purple.
For example, in \autoref{fig:idealised-wires}, we consider a submodel of circuit-level noise that idealises three edges, the $ZZ$ measurement as fault-tolerant, and the $XZ$ measurement as fault-free.

\subsection{Detecting and Correcting Faults}
Fault-tolerant quantum computing involves encoding circuits with redundancy so that faults can be detected and corrected.
This includes performing Pauli measurements that, in an ideal fault-free scenario, have predetermined outcomes.
When their outcomes differ from expectations, it indicates that something went wrong.

Following \textcite{gidneyStimFast2021}, we define:
\begin{definition}[Detecting sets]
    \label{def:detecting-sets}
 A \emph{detecting set} is a set of measurements whose parity of the measurement outcomes in the fault-free case is predetermined. 
\end{definition}

The most trivial example of a detecting set is two consecutive $ZZ$ Pauli measurements (see \autoref{fig:detecting-sets-examples} (a)).
In the fault-free case, we would expect these two measurements to have the same outcome, i.e.\@ their parity to be even.
If an anti-commuting fault, i.e.\@ a fault of type $X$ or $Y$, occurs that flips only one of the measurements, we detect it.
Similarly, flags~\parencite{chamberlandFlagFaulttolerant2018} can be understood as introducing additional detecting sets to a circuit, where we expect the outcome of the flag to be $0$, introducing a detecting set of size one (see \autoref{fig:detecting-sets-examples} (b)).
\begin{figure}
    \centering
    \begin{minipage}{0.45\textwidth}
        \centering
        \input{./figures/02-errors/detecting-set-example-1.tikz}
        \caption*{(a)}
    \end{minipage}
    \begin{minipage}{0.45\textwidth}
        \centering
        \input{./figures/02-errors/detecting-set-example-2.tikz}
        \caption*{(b)}
    \end{minipage}
    \caption{Two examples of detecting sets. Both detect if there are an odd number of faults of type $X$ and $Y$ in the highlighted spacetime locations. (a) Two consecutive $ZZ$ measurements. (b) A flag around two CNOTs.}
    \label{fig:detecting-sets-examples}
\end{figure}

We say: 
\begin{definition}[Detectability of a fault]
 A fault is \emph{detectable} if it flips an odd number of measurements in a detecting set.
\end{definition}

By flipping an odd number of measurements in a detecting set, a detectable fault creates a measurement outcome that violates the expected parity of said detecting set.
Importantly, a fault may be undetectable within the analysed circuit fragment while being detectable by a later circuit fragment. 
For example, faults on the outputs of a circuit are undetectable --- they do not flip any measurements --- but might well be detected by the next part of the circuit. 

To calculate the detecting sets of a Clifford circuit and which faults they can detect, one can use the mathematical tools laid out by \textcite{delfosseSpacetimeCodes2023}, building on the stabiliser formalism.
Alternatively, within the ZX calculus, detecting regions and detectable faults can be formalised in terms of Pauli webs~\parencite{bombinUnifyingFlavorsFault2024,rodatzFloquetifyingStabiliser2024}.
As we review the ZX calculus and Pauli webs in the second half of this paper, we defer the technical details of calculating detection sets and the detectability of faults to \autoref{sec:detection-regions}. 

We can quantify a circuit's resilience to noise by generalizing the notion of distance from error correction codes to circuits.
Similar to~\parencite{bravyiHighthresholdLowoverhead2024}, we define:
 
\begin{definition}[Circuit Distance]
 Let $C$ be a circuit under some noise model $\mathcal F$.
 Then the \emph{distance} of $C$ under $\mathcal F$, denoted as $dist_{\mathcal F}(C)$, is defined as the minimum weight of any non-trivial, undetectable fault in $\langle \mathcal F \rangle$.
\end{definition}

Relating this to the prior work of \textcite{delfosseSpacetimeCodes2023}, we observe that the circuit distance is closely related to the distance of the circuit's spacetime code.

As faults on the boundary edges are undetectable, for a circuit to have a non-trivial distance, it cannot have any faults affecting inputs or outputs.
Just as distance and correctability are closely related in quantum error correction codes, we observe: 

\begin{proposition}[Correctability of circuits]
    \label{prop:correctability-of-circuits}
 Let $C$ be a circuit with distance $d$ under some noise model $\mathcal F$. 
 Then any two non-equivalent faults of weight at most $\lfloor \frac{d - 1}{2} \rfloor$ have different syndromes.
\end{proposition}
\begin{proof}
 We can prove this by contradiction. 
 Let us assume $F_1, F_2$ are non-equivalent faults of weight at most $\lfloor \frac{d - 1}{2} \rfloor$ with the same syndrome.
 Then the combined fault $F_1 F_2$ is undetectable, as the syndrome of the combination of two faults is the same as adding the syndromes of the individual faults mod 2~\parencite{delfosseSpacetimeCodes2023}.
 Additionally, $F_1 F_2$ is non-trivial, as $F_1$ and $F_2$ were assumed not to be equivalent.
 But then $F_1 F_2$ is a non-trivial, undetectable fault of weight at most $\lfloor \frac{d - 1}{2} \rfloor \times 2 = d - 1$.
 This contradicts our assumption that the circuit has distance $d$. 
\end{proof}

This means that equivalence classes of faults that contain elements of weight at most $\lfloor \frac{d - 1}{2} \rfloor$ are uniquely identifiable.
For each such equivalence class of faults, we can calculate its effect and the corresponding correction~\parencite{delfosseSpacetimeCodes2023}.
Just as with quantum error detection codes, identifying the correct equivalence class of faults, given the syndrome information, is computationally hard. 
\autoref{prop:correctability-of-circuits} purely guarantees that this is theoretically doable, for example, using a look-up table.
Leveraging this close relationship between distance and correctability, for the remainder of this paper, we focus on the detectability of faults.

\section{Fault Equivalence}
\label{sec:circuit-fault-equivalence}
The central focus of this paper is to study the behaviour of circuits under noise.
Specifically, we observe that circuits which implement the same operation can nonetheless behave very differently when faults occur, with faults in one potentially being far more detrimental than in another.
Therefore, when synthesising or optimising circuits in a noisy setting, the equivalence of noise-free behaviour alone is insufficient.
We need a more fine-grained notion of equivalence, which we call \emph{fault equivalence}, relating circuits under noise.
In the remainder of this section, we argue that fault equivalence is not only a useful notion but one implicitly present in much of the thinking around fault-tolerant circuits.
We illustrate this by reframing various fault-tolerant circuit compilation tasks from the literature as synthesising implementable circuits that are fault-equivalent to some idealised specification of a circuit and its behaviour. 

\subsection{Behaviour of Circuits under Noise}
To showcase the differing behaviour of circuits under noise, consider the following two circuits, both preparing the cat state on four qubits, which is defined, up to normalisation, as $\ket{0000} + \ket{1111}$:
\[\input{./figures/03-fault-equivalence-in-action/example-faulty-cat-state.tikz} \qquad \qquad \text{vs. } \qquad \qquad \input{./figures/03-fault-equivalence-in-action/example-fault-free-cat-state.tikz}\]

The circuit on the right also prepares a cat state using the gates available on a quantum computer. 
The circuit on the right consists of an idealised, fault-free cat state preparation.
The only allowed faults are Pauli flips on the data qubits.
These two circuits behave differently under noise.
Notably, certain faults in this circuit can propagate and generate two faults on the data qubits:
\[\input{./figures/03-fault-equivalence-in-action/example-faulty-cat-state-fault-1.tikz} \quad = \quad  \input{./figures/03-fault-equivalence-in-action/example-faulty-cat-state-fault-2.tikz} \quad = \quad  \input{./figures/03-fault-equivalence-in-action/example-fault-free-cat-state-fault.tikz}\]

A single fault event on the larger circuit can be as detrimental as two fault events on the idealised circuit. 
This is problematic.
For example, consider a circuit that uses the idealised cat state preparation and has a distance of two, meaning all weight-one faults are detectable.
If we replace the ideal preparation with this implementation, a single fault can propagate and become a weight-two fault.
This fault may no longer be detectable, potentially reducing the circuit's distance to one.
Thus, despite behaving the same in the noise-free setting, these two circuits are not equivalent in a noisy setting.

We define:
\begin{definition}[Fault equivalence, $w$-fault equivalence]
   \label{def:fault-equivalence}
 Let $C_1, C_2$ be two circuits that implement the same linear map with respective noise models $\mathcal{F}_1, \mathcal{F}_2$. The circuit $C_1$ under $\mathcal{F}_1$ is \emph{$w$-fault-equivalent} to $C_2$ under $\mathcal{F}_2$ if and only if for all faults $F_1 \in \langle \mathcal{F}_1 \rangle$ with weight $wt(F_1) < w$, we have either:
    \begin{enumerate}
        \item $F_1$ is detectable, or
        \item there exists a fault $F_2 \in \langle \mathcal{F}_2 \rangle$ on $C_2$ such that:
        \begin{itemize}
          \item $wt({F}_2) \leq wt({F}_1)$ and
          \item $C_1^{F_1} = C_2^{F_2}$.
        \end{itemize}
      \end{enumerate}
 The condition must similarly hold for all faults ${F}_2 \in \langle \mathcal{F}_2 \rangle$ with weight $wt({F}_2) < w$, making this equivalence relation symmetric.
 We write $C_1 \underset{w}{\FaultEq} C_2$.
 Two circuits $C_1$ and $C_2$ are \emph{fault-equivalent} if they are $w$-fault-equivalent for all $w \in \mathbb{N}$.
 We write $C_1 \FaultEq C_2$.
\end{definition}

Intuitively, $w$-fault equivalence guarantees a correspondence between the undetectable faults of weight less than $w$ in two circuits.
It ensures that for any undetectable fault on one side, there exists an equivalent fault on the other side with at most the same weight.
Therefore, neither circuit can achieve a specific undetectable fault with fewer fault events than the other.
Similarly, fault equivalence guarantees that every fault, no matter the weight, is either detectable or has a corresponding fault of at most the same weight on the other side. 

An alternative, equivalent formulation of fault-equivalence is that we can observe that undetectable faults $F$ that create the same linear map $D^F$ form an equivalence class up to stabilisers of the gates, also known as gauges \cite{ruschCompletenessFault2025}. 
Fault-equivalence requires that the minimum weight representative in each equivalence class has the same weight on $C_1$ as on $C_2$. 

This means that ($w$-)fault equivalence is a more fine-grained notion of equivalence than the usual notion of circuit equivalence, which requires the two circuits to represent the same linear map. 
Fault equivalence guarantees that two circuits represent the same linear map. 
However, the other way around, two circuits implementing the same linear map does not guarantee that they are fault-equivalent. 

Coming back to the counter-example given above, we can see that, despite implementing the same linear map, the two circuits are not fault-equivalent; there exists no fault of weight one on the idealised cat state that has the same effect as the fault shown above. 
The smallest fault with the same effect is of weight two.

\subsection{Defining Fault-Tolerant Circuit Synthesis with Fault Equivalence}
To reason that fault equivalence is a useful notion, we recover the description of common problems in FTQC in terms of fault equivalence, focusing on fault-tolerant circuit synthesis.
Essentially, they boil down to synthesising implementable circuits that are fault-equivalent to some idealised specification of a circuit.

For example, fault-tolerant cat state preparation up to some weight $w$ is usually expressed as:
``Find a circuit that implements the cat state such that all faults of weight $n < w$ are either detectable or create at most $n$ data errors.''
This is formally equivalent to: 

\begin{restatable}[Fault-tolerant cat state preparation]{definition}{ftCatStatePrep}
  \label{def:ft-cat-state-prep}
 A circuit $C$ \emph{fault-tolerantly prepares the cat state} up to some weight $w$, if:
  \[
 \input{./figures/03-fault-equivalence-in-action/def-ft-cat-state.tikz}
  \]
\end{restatable}

Here, as well as in the rest of the paper, we consider both circuits under circuit-level noise and idealisations thereof. 
From now on, we will also omit drawing the spacetime locations.

If $C$ is $w$-fault-equivalent to the specification, we are guaranteed that any circuit-level fault $F$ on $C$ of weight less than $w$ must either be detectable or correspond to a fault of at most the same weight on the idealised specification.
As the only allowed faults on the idealised specification exactly correspond to data errors, this is formally equivalent to the definition usually found in the literature.

Similarly, fault-tolerant logical state preparation can be expressed as:
\begin{definition}[Fault-tolerant logical state preparation]
 Let $enc$ be the encoder of some error correction code and $\psi$ be a logical state.
 Then a circuit $C$ \emph{fault-tolerantly prepares a $\psi$} up to some weight $w$, if:
    \[\input{./figures/03-fault-equivalence-in-action/def-ft-state.tikz}\]
\end{definition}

Given that we idealised the state preparation on the logicals and the encoding as fault-free, the only faults we allow on this idealised specification are data errors after a successful, fault-free state preparation. 
Once again, this corresponds to the more common definition of logical state preparation as: ``A circuit that prepares the logical state such that any fault of weight $n < w$ is either detectable or creates at most $n$ many faults on the data.''

Going beyond state preparation, we can equally express fault-tolerant syndrome extraction as:
\begin{restatable}[Fault-tolerant syndrome extraction]{definition}{ftSyndromeExtraction}
  \label{def:ft-syndrome-extraction}
 Given a stabiliser code $enc$ of distance $d$ with stabiliser generators $S_1, \dots S_m$, a circuit $C$ \emph{fault-tolerantly implements the syndrome extraction} for $enc$, if:
  \[\input{./figures/03-fault-equivalence-in-action/ideal-error-detection.tikz}\]
\end{restatable}
Here, we require that, given a stabiliser code, we want a circuit that is $d$-fault-equivalent to measuring a generating set of the stabilisers perfectly without any fault in between the measurements. 

Oftentimes, syndrome extraction is done by repeating imperfect, yet fault-tolerant syndrome measurements, meaning syndrome measurements may give the wrong result if measurement errors occur, however, they do not propagate faults badly:
\begin{definition}[Fault-tolerant syndrome measurement]
    \label{def:ft-syndrome-measurement}
 A circuit $C$ \emph{fault-tolerantly implements a syndrome measurement} of type $P_1P_2 \dots P_n$ up to some weight $w$, if:
    \[\input{./figures/03-fault-equivalence-in-action/ZZZZ-parity-check.tikz}
    \]
\end{definition}

As atomic faults on fault-tolerant Pauli measurements consist of either data errors and/or measurement flips, this definition requires that any fault of weight $n$ is either detectable or creates at most $n$ data errors and/or a flip of the measurement outcome. 
For syndrome measurements in particular, \textcite[Definition 15]{baconSparseQuantum2017} give a formally equivalent notion which they call \emph{fault-tolerant sub-projection}, stating that any fault should either make the projection go to $0$, i.e.\@ be detectable, or have a corresponding Pauli error on the inputs and outputs of at most the same weight.
As such, fault equivalence not only captures existing notions in the literature but has previously been proposed for very specific settings.
Here, we show how it can be generalised beyond fault-tolerant syndrome measurements. 

\begin{remark}
 The syndrome measurement with one bare ancilla used in the surface code is not considered fault-tolerant under our definition above. 
 This is due to the fact that a circuit that satisfies the condition above can be used in any context, while the syndrome extraction of the surface code is bespoke to that specific context. 
 It is possible within the language of fault equivalence to further specify context and thus allow certain hook errors in the specification of one's desired gadgets. 
 However, that is beyond the scope of the current work.    
\end{remark}

\subsection{Formalising FTQC with Fault Equivalence}
Beyond using fault equivalence to reason about fault-tolerant gadgets, we can use it to recover notions of quantum error correction and fault-tolerant quantum computing.

For example, we can observe a close link between circuit distance and fault-equivalence:
\begin{proposition}
    \label{prop:dist-as-fe}
 Let $C$ be a quantum circuit with some noise model $\mathcal F$.
 Then $dist_{\mathcal F}(C) = d$, i.e.\@ $C$ has circuit distance $d$, if and only if $d$ is the maximum value such that $C$ is $d$-fault-equivalent to its idealised, fault-free implementation, i.e.:
    \[\resizebox{\textwidth}{!}{\input{./figures/03-fault-equivalence-in-action/dist-as-fe.tikz}}\]
\end{proposition}
\begin{proof}
  $\Rightarrow$: If $C$ has $dist_{\mathcal F}(C) = d$ that means that all non-trivial faults of weight less than $d$ must be detectable. But then, clearly, $C$ is $d$-fault-equivalent to its idealised fault-free version.
 However, as $C$ has distance $d$ and not distance $d + 1$, there must exist some fault $F$ of weight $d$ that is non-trivial and undetectable, so $C$ is not $d+1$-fault-equivalent to its fault-free self.

  $\Leftarrow$: If $C$ is $d$-fault-equivalent to its idealised fault-free version but not $d + 1$-fault-equivalent to its idealised fault-free version, that means that all faults of weight less than $d$ must be detectable or trivial. However, there exists a fault of weight less than $d + 1$ that is non-trivial and undetectable. Therefore, $dist_{\mathcal F}(C) = d$.
\end{proof}

Furthermore, we can use fault equivalence to recover the standard definition of a quantum error correction code. 
A usual definition of a quantum error correction code is:

\begin{definition}[Quantum Error Correction --- Standard Definition]
\label{def:qec-lit}
A linear map $enc: (\mathbb{C}^{2})^{\otimes k} \to (\mathbb{C}^{2})^{\otimes n}$ is an encoder for a \emph{quantum error correction code} with distance $d$ if there exists a linear map $dec: (\mathbb{C}^{2})^{\otimes n} \to (\mathbb{C}^{2})^{\otimes k}$ (the decoder) such that for all faults ${F} \in \overline{\mathcal{P}^n}$ with weight $wt({F}) < d$ applied to the encoded state, the fault ${F}$ is detectable by the decoder.
\end{definition}

This concept can be expressed equivalently using fault equivalence:

\begin{definition}[Quantum Error Correction  --- Using Fault Equivalence]
\label{def:qec}
A linear map $enc: (\mathbb{C}^{2})^{\otimes k} \to (\mathbb{C}^{2})^{\otimes n}$ is an encoder for a \emph{quantum error correction code} with distance $d$ if there exists a linear map $dec: (\mathbb{C}^{2})^{\otimes n} \to (\mathbb{C}^{2})^{\otimes k}$ (the decoder) such that:
\[\input{./figures/03-fault-equivalence-in-action/def-qec-dist-preserv.tikz}\]
\end{definition}

The fault equivalence condition requires that any fault of weight less than $d$ on the left-hand side has to either be detectable or has to have an equivalent fault on the right-hand side.
As the right-hand side is idealized as fault-free, the only fault it allows for is the trivial, empty fault.
This means that any fault of weight less than $d$ either has to be detectable or has to be equivalent to the trivial fault.
This is exactly the usual definition of the detectability of faults.

Going one step further, we can define a notion of fault-tolerant quantum computing in terms of fault equivalence. 
Quantum error correction is primarily about protecting quantum data (states) from noise by encoding it with redundancy such that errors can be detected and corrected.
Similarly, fault-tolerant quantum computing aims to protect quantum computations (circuits) from faults such that faults can be detected and corrected.

We define:
\begin{definition}[Fault-tolerant quantum computation under adversarial noise]
 Given a quantum circuit $C$, $C_{FT}$ is a $d$-fault-tolerant implementation of $C$ if we have:
        \[\input{./figures/03-fault-equivalence-in-action/def-fault-tolerance.tikz}\]
\end{definition}

That is, given an idealized quantum circuit $C$ that we want to execute, we have to find an encoding of that circuit $C_{FT}$ that performs the same computation with redundancy in such a way that faults can be detected and corrected.
Similar to our definition of quantum error correction, this definition requires that any fault of weight less than $d$ on $C_{FT}$ has to either be detectable or trivial. 
Thus, we have encoded $C$ in such a way that faults of weight less than $d$ must be detectable and therefore, by \autoref{prop:correctability-of-circuits}, faults of weight at most $\lfloor \frac{d-1}{2} \rfloor$ must be correctable.

\subsection{Properties of Fault Equivalence}
Fault equivalence is a useful property for reasoning about quantum circuits in a noisy setting.
However, checking whether two quantum circuits are $w$-fault-equivalent is computationally hard. 
Intuitively, we have to iterate all fault combinations of weight less than $w$ to check whether they are detectable or have a corresponding fault in the other diagram. 
More formally, we can observe that calculating the distance of a code can be reduced to checking fault equivalence:
\begin{theorem}
 Checking whether two quantum circuits are fault-equivalent is NP-hard.
\end{theorem}
\begin{proof}
 We will prove this by showing that calculating the distance of a code is the same as calculating the maximum $w$ for which two specific circuits are fault-equivalent.

 Given a stabiliser code with distance $d$ with stabiliser generators $S_1, \dots S_m$, we have:
    \[dist\left(\input{./figures/03-fault-equivalence-in-action/NP-hard-proof.tikz} \right) = d\]
 This is due to the fact that when perfectly measuring the stabilisers of the code, then allowing faults and then measuring the stabilisers of the code again, the smallest non-trivial, undetectable error must be of weight exactly $d$.
    
 But then, as calculating the distance of a diagram can be phrased in terms of checking fault equivalence (\autoref{prop:dist-as-fe}), calculating fault equivalence must be at least as hard as calculating the distance of a stabiliser code, which is known to be NP-hard~\parencite{kapshikarHardnessMinimum2023}.
\end{proof}

However, while arbitrary fault equivalence is computationally hard to check, there are some properties of fault equivalence that will be very useful throughout the rest of this paper:

\begin{proposition}
  \label{prop:fault-equivalence-compositionality}
 Fault equivalence is compositional, i.e.:
    \[\resizebox{\textwidth}{!}{\input{./figures/03-fault-equivalence-in-action/fe-compositionality.tikz}}\]
\end{proposition}
\begin{proof}
 First, we will prove the compositionality of fault equivalence under parallel composition.
 To show that $C_1 \otimes C_2 \underset{\scriptstyle \min(w_1, w_2)}{\FaultEq} C'_1 \otimes C'_2$, we have to show that any undetectable fault of weight less than $\min(w_1, w_2)$ on $C_1 \otimes C_2$ has a corresponding fault on $C'_1 \otimes C'_2$ and vice versa.

 Let $F$ be an undetectable fault on $C_1 \otimes C_2$ of weight less than $\min(w_1, w_2)$.
 We can separate $F = F_1F_2$ into components $F_1, F_2$ respectively acting only on $C_1$ and $C_2$.
 As they act on disjoint circuits, we have $wt(F) = wt(F_1) + wt(F_2)$.
 But then, as $wt(F_1) \leq wt(F)$, by $C_1 \underset{\scriptstyle w_1}{\FaultEq} C'_1$, we know there exists a fault $F'_1$ such that $C_1^{F_1} = {C_1'}^{F'_1}$ and that $wt(F'_1) \leq wt(F_1)$.
 Similarly, for $F_2$, we can find an $F'_2$ such that $C_2^{F_2} = {C_2'}^{F'_2}$ and that $wt(F'_2) \leq wt(F_2)$.
 But then $F' = F'_1F'_2$ is a fault that acts on $C'_1 \otimes C'_2$ for which we have $wt(F') \leq wt(F'_1) + wt(F'_2) \leq wt(F_1) + wt(F_2) = wt(F)$ and $C_1^{F_1} \otimes C_2^{F_2} = {C_1'}^{F'_1} \otimes {C'_2}^{F'_2}$.
 Thus, for all undetectable faults of weight less than $\min(w_1, w_2)$ on $C_1 \otimes C_2$, we can find a corresponding fault on $C'_1 \otimes C'_2$.
 The other direction holds analogously.

 The compositionality of fault equivalence under sequential composition can be proven analogously. 
 Here, one has to take slightly more care when splitting the fault $F$ into components $F_1$ and $F_2$, respectively acting on $C_1$ and $C_2$. 
 In particular, when the fault acts on the interface between the two circuits, one has to split $F$ in such that $wt(F_1) + wt(F_2) \leq wt(F)$. 
 As the set of atomic faults on $C_2 \circ C_1$ is exactly the union of the atomic faults on $C_1$ and $C_2$, this is always possible. 
\end{proof}

This means that, if we can break two circuits into smaller, fault-equivalent components, we can simplify checking fault equivalence.

Furthermore, we can state:
\begin{proposition}
 Fault equivalence is transitive, i.e.:
    \[\input{./figures/03-fault-equivalence-in-action/fe-transitivity.tikz}\]
\end{proposition}
\begin{proof}
 To show that $C_1 \underset{\scriptstyle \min(w_1, w_2)}{\FaultEq} C_3$, we have to show that any undetectable fault of weight less than $\min(w_1, w_2)$ on $C_1$ has a corresponding fault on $C_3$ and vice versa.

 Let $F_1$ be an undetectable fault on $C_1$ of weight less than $\min(w_1, w_2)$.
 But then, as $wt(F_1) < w_1$, by $C_1 \underset{\scriptstyle w_1}{\FaultEq} C_2$, we know that there exists a fault $F_2$ such that $C_1^{F_1} = C_2^{F_2}$ and that $wt(F_2) \leq wt(F_1)$.
 But then, as $wt(F_2) \leq wt(F_1) <  \min(w_1, w_2)$, by $C_2 \underset{\scriptstyle w_2}{\FaultEq} C_3$, we know that there exists a fault $F_3$ such that $C_2^{F_2} = C_3^{F_3}$ and that $wt(F_3) \leq wt(F_2)$.
 But then we have a fault $F_3$ that satisfies $wt(F_3) \leq wt(F_2) \leq wt(F_1)$ and $ C_3^{F_3} = C_2^{F_2} = C_1^{F_1}$ and therefore, we have found a corresponding fault for $F_1$.

 Showing that for all faults on $C_3$ there exist corresponding faults in $C_1$ can be shown analogously.
 But then we have shown that $C_1 \underset{\scriptstyle \min(w_1, w_2)}{{\FaultEq}} C_3$.
\end{proof}

As we have previously expressed distance as a property in terms of fault equivalence, we now have:
\begin{corollary}
    \label{cor:fe-preserves-dist}
 Let $C_1, C_2$ be quantum circuits such that $C_1 \underset{\scriptstyle w}{\FaultEq} C_2$ under some noise models $\mathcal{F}_1, \mathcal{F}_2$ and $dist_{\mathcal{F}_1}(C_1) = d$. Then we have $dist_{\mathcal{F}_2}(C_2) \geq \min(w, d)$.
\end{corollary}
\begin{proof}
 Since $C_1$ and $C_2$ are $w$-fault-equivalent this means that the idealised, fault-free circuits of $C_1$ and $C_2$ must be equivalent, as $C_1$ under the trivial fault must be equivalent to $C_2$ under the trivial fault.
 But then we have:
    \[\input{./figures/03-fault-equivalence-in-action/distance-transitivity-proof.tikz}\]
 But then, by transitivity of fault equivalence, $C_2$ is $\min(w, d)$-fault-equivalent to its idealised fault-free version and thus we have a lower bound on $dist_{\mathcal{F}_2}(C_2)$.
\end{proof}

This means that manipulating circuits while preserving fault equivalence preserves their distance, and, thus, the correctability of faults.
In fact, in \cite{rodatzFloquetifyingStabiliser2024}, we defined the notion of distance-preserving rewrites, which we defined to be the rewrites that preserve the circuit distance in all contexts. 
We showed that fault-equivalent rewrites and distance-preserving rewrites are formally equivalent. 

Fault equivalence will be the key ingredient in our procedure for fault tolerance by construction; if we start with an idealised circuit that behaves well under noise and bring it into the shape of a fault-equivalent, non-idealised circuit that we can execute on hardware, we are guaranteed that the distance of the resulting circuit will be the same as the distance of our original circuit and therefore fault of weight less than the distance to be detectable.

\part{Faults on ZX Diagrams}

\section{Intro to ZX Calculus}
Having argued that fault equivalence is a useful concept for reasoning about circuits in a noisy setting, we turn our attention towards manipulating quantum circuits while preserving fault equivalence.
This allows us to present novel ways of solving the fault-tolerant circuit synthesis tasks stated above.
To do this, we utilise the ZX calculus.

The ZX calculus~\parencite{coeckeInteractingQuantumObservables2008} is a graphical formalism for representing and reasoning about quantum circuits.
While the ZX calculus is universal for arbitrary linear maps between qubits, we focus on the Clifford fragment of the ZX calculus, which is universal for stabiliser circuits~\parencite{backensZXcalculusComplete2014}.
For a comprehensive overview of the complete calculus, readers are referred to~\parencite{vandeweteringZXcalculusWorkingQuantum2020}.
Additionally, there exist variations of the ZX calculus for qudits~\parencite{poorZXcalculusComplete2024} and continuous-variable quantum computation~\parencite{shaikhFockedupZX2024}.

\subsection{ZX Diagrams}\label{sec:zx-diagrams}

A ZX diagram is a graphical representation of a linear map, composed of interconnected nodes with parameters and colours.
Nodes are referred to as \enquote{spiders}, of which there are two types: Z-spiders and X-spiders (hence ZX).
Each spider carries a phase parameter $\alpha \in \{0, \frac{\pi}{2}, \pi, \frac{3\pi}{2}\}$ and has an arbitrary number of \enquote{legs} or connections, which represent its input/output qubit ports.

The most basic ZX diagrams consist of one spider, respectively defined as: 
\begin{definition}[Z spider, X spider]
  \begin{align*}
 \textit{Z-spider:} \qquad
 \input{./figures/04-zx/ZX/z-spider.tikz}
 \quad &\coloneqq \quad
 \ket{0}^{\otimes n}\! \bra{0}^{\otimes m} + e^{i k\frac{\pi}{2}} \ket{1}^{\otimes n}\! \bra{1}^{\otimes m}\\[8pt]
 \textit{X-spider:} \qquad
 \input{./figures/04-zx/ZX/x-spider.tikz}
 \quad &\coloneqq \quad
 \ket{+}^{\otimes n}\! \bra{+}^{\otimes m} + e^{i k\frac{\pi}{2}} \ket{-}^{\otimes n}\! \bra{-}^{\otimes m}
  \end{align*}
\end{definition}
By convention, if the phase parameter is not explicitly annotated on a spider, it defaults to $0$. We call spiders with a phase of $0$ \textit{phase-free}.
In addition to spiders, we allow \textit{identities}, \textit{swaps}, \textit{cups}, and \textit{caps} in ZX diagrams, which are defined as follows:
\[
\tikz[tikzfig]{ \draw (0,0) -- (2,0); } \ \ :=\ \  \sum_i |i\>\<i|
\qquad
\input{./figures/swap.tikz} \ \ :=\ \ \sum_{ij} |ij\>\<ji|
\qquad
\begin{tikzpicture}
	\begin{pgfonlayer}{nodelayer}
		\node [style=none] (12) at (0, -0.75) {};
		\node [style=none] (14) at (0, 0.75) {};
		\node [style=none] (15) at (0, 0.75) {};
		\node [style=none] (16) at (0, -0.75) {};
	\end{pgfonlayer}
	\begin{pgfonlayer}{edgelayer}
		\draw [bend left=90, looseness=1.75] (16.center) to (15.center);
	\end{pgfonlayer}
\end{tikzpicture}
 \ \ := \ \ \sum_i |ii\>
\qquad
\begin{tikzpicture}
	\begin{pgfonlayer}{nodelayer}
		\node [style=none] (12) at (0, -0.75) {};
		\node [style=none] (14) at (0, 0.75) {};
		\node [style=none] (15) at (0, 0.75) {};
		\node [style=none] (16) at (0, -0.75) {};
	\end{pgfonlayer}
	\begin{pgfonlayer}{edgelayer}
		\draw [bend right=90, looseness=1.75] (16.center) to (15.center);
	\end{pgfonlayer}
\end{tikzpicture}
 \ \ := \ \ \sum_i \<ii|
\]
We also use the Hadamard gate, represented as a small box, as a basic component of ZX diagrams, although it can be constructed (up to a global phase) in terms of Z- and X-spiders:
\[
\input{./figures/04-zx/h-def.tikz} \ \ =\ \ |{+}\>\<0| + |{-}\>\<1|
\]

Larger, more complex diagrams can be created by composing smaller diagrams: 
\begin{definition}[Sequential composition, parallel composition] 
 Let $D_1: (\mathbb{C}^{2})^{\otimes n} \to (\mathbb{C}^{2})^{\otimes m}, D_2: (\mathbb{C}^{2})^{\otimes m} \to (\mathbb{C}^{2})^{\otimes k}$ be ZX diagrams. Then we can sequentially compose $D_2 \circ D_1: (\mathbb{C}^{2})^{\otimes n} \to (\mathbb{C}^{2})^{\otimes k}$ as: 
  \[
 \interp{\input{./figures/04-zx/ZX/D2-circ-D1.tikz}}
 \quad \coloneqq \quad
 \interp{\input{./figures/04-zx/ZX/D2.tikz}} \circ \interp{\input{./figures/04-zx/ZX/D1.tikz}}
  \]
 Let $D_1: (\mathbb{C}^{2})^{\otimes n} \to (\mathbb{C}^{2})^{\otimes m}, D_2: (\mathbb{C}^{2})^{\otimes k} \to (\mathbb{C}^{2})^{\otimes l}$ be ZX diagrams. Then we can parallelly compose $D_1 \otimes D_2: (\mathbb{C}^{2})^{\otimes (n + k)} \to (\mathbb{C}^{2})^{\otimes (m + l)}$ as: 
  \[
 \interp{\input{./figures/04-zx/ZX/D1-tensor-D3.tikz}}
 \quad \coloneqq \quad
 \interp{\input{./figures/04-zx/ZX/D1.tikz}} \otimes \interp{\input{./figures/04-zx/ZX/D3.tikz}}
  \]
\end{definition}

Using the basic spiders and the two methods of composition, a key characteristic of the ZX calculus is that it is universal: 
\begin{theorem}
 Any linear map $L: (\mathbb{C}^{2})^{\otimes n} \to (\mathbb{C}^{2})^{\otimes m}$ can be represented as a ZX diagram.
\end{theorem}
\begin{proof}
 See \textcite[Proposition 2.18]{coeckeInteractingQuantumObservables2008}.
\end{proof}

To interpret larger diagrams, one can deconstruct them into their basic components:
\[\scalebox{0.9}{\input{./figures/04-zx/ZX/composition-example.tikz}}\]
While every ZX diagram can be interpreted as its underlying linear map this way, in practice, this is not usually necessary. 
Instead, we translate directly between quantum circuits and ZX diagrams. 
\autoref{fig:mappings} illustrates how various quantum states and Clifford gates are represented as ZX diagrams.
To construct a ZX diagram for a larger circuit, one can translate the individual gates and compose them according to the structure of the circuit.
\begin{figure}[htb]
  \begin{minipage}{.32\textwidth}
    \begin{align*}
 \ket{0}
 \quad =& \quad
 \input{./figures/04-zx/ZX/x-state-zero.tikz} \\[8pt]
 \ket{1}
 \quad =& \quad
 \input{./figures/04-zx/ZX/x-state-one.tikz} \\[8pt]
 \ket{+}
 \quad =& \quad
 \input{./figures/04-zx/ZX/z-state-zero.tikz} \\[8pt]
 \ket{-}
 \quad =& \quad
 \input{./figures/04-zx/ZX/z-state-one.tikz} \\[8pt]
 \input{./figures/04-zx/QC/not.tikz}
 \quad =& \quad
 \input{./figures/04-zx/ZX/not.tikz} \\[8pt]
 \input{./figures/04-zx/QC/y.tikz}
 \quad =& \quad
 \input{./figures/04-zx/ZX/y.tikz}
\end{align*}
\end{minipage}
\begin{minipage}{.32\textwidth}
  \begin{align*}
 \input{./figures/04-zx/QC/s.tikz}
 \quad =& \quad
 \input{./figures/04-zx/ZX/s.tikz}\\[8pt]
 \input{./figures/04-zx/QC/cnot.tikz}
 \quad =& \quad
 \input{./figures/04-zx/ZX/cnot.tikz} \\[8pt]
 \input{./figures/04-zx/QC/cz.tikz}
 \quad =& \quad
 \input{./figures/04-zx/ZX/cz.tikz} \\[8pt]
 \input{./figures/z-meas.tikz} \quad =& \quad \input{./figures/z-meas-zx.tikz} \\[8pt]
 \input{./figures/x-meas.tikz} \quad =& \quad \input{./figures/x-meas-zx.tikz}
    \end{align*}
  \end{minipage}
  \begin{minipage}{.32\textwidth}
    \begin{align*}
 \input{./figures/04-zx/QC/cat-state-4.tikz}
 \quad =& \quad
 \input{./figures/04-zx/ZX/cat-state-4.tikz}\\[8pt]
 \input{./figures/04-zx/QC/z-dots-z-parity-check.tikz}
 \quad =& \quad
 \input{./figures/04-zx/ZX/z-dots-z-parity-check.tikz}
    \end{align*}
  \end{minipage}
  \caption{Mapping quantum circuits representing common states and operations to ZX diagrams. As we note in the next section, the direction of wires is irrelevant, so we can write CNOT gates and CZ gates using vertical wires between spiders without ambiguity.}
  \label{fig:mappings}
\end{figure}

\subsection{Rewrite Rules}\label{sec:zx-rules}
\label{sec:zx-rewrites}
Beyond representing quantum circuits, the distinguishing feature of the ZX calculus lies in its ability to relate different quantum circuits.
It comes with a set of graphical rewrite rules that allow ZX diagrams to be transformed into equivalent diagrams (i.e.\@ representing the same linear map) by applying local transformations within a larger diagram.
This enables the manipulation, simplification, and verification of quantum computations directly in the graphical language.

First and foremost, a fundamental rewrite of the ZX calculus, often assumed implicitly, is that \emph{Only Connectivity Matters} (\TextOCM).
This principle allows us to rearrange spiders and bend their legs, provided the overall connectivity and the positions of inputs/outputs of the diagram are preserved.
For example, the following diagrams are all considered equivalent under OCM:
\[
 \input{./figures/04-zx/OCM.tikz}
\]
Beyond OCM, there are eight additional rewrite rules of the ZX calculus, from which complex transformations can be derived.
\autoref{fig:axioms} illustrates a set of axioms for the Clifford fragment of the ZX calculus (up to global scalars).
This set of axioms is known to be complete~\parencite{backens2017AsimplifiedStabiliserZXCalculus}, meaning that if two diagrams implement the same linear map, there is guaranteed to be a set of rewrites from this set of axioms that transforms one into the other. 
This gives the calculus a deductive power equal to matrix algebra but within an intuitive graphical framework.
\begin{figure}[htb]
  \[
 \input{./figures/04-zx/axioms.tikz}
  \]
  \caption{A complete set of axioms for the Clifford fragment of the ZX calculus. All rules also hold in the colour inverse. Equalities shown hold up to scalars.}
  \label{fig:axioms}
\end{figure}

The utility of the ZX calculus becomes clearer when we examine how individual rewrite rules encode familiar concepts in quantum computing or aid in specific procedures. 
For example, the ZX calculus can be used to propagate Pauli faults through Clifford circuits. 
Consider a $Y$ fault before the control of a CNOT:
\[
 \input{./figures/04-zx/pauli-propagation.tikz}
\]
When mapping between quantum circuits and ZX diagrams, we highlight the corresponding components to aid readability. 
To map the original circuit to a ZX diagram, we express the $Y$ fault as a combination of an $X$ and a $Z$ fault.
We can do this unambiguously, as we only consider diagrams up to a global phase.
Specifically, this means that we also do not have to care about the order of the decomposition, i.e.\@ whether we decompose the $Y$ into an $X$ and a $Z$ or a $Z$ and an $X$.
In the first rewrite, we fuse the two Z spiders, adding their respective phases --- by convention, we consider an empty spider to have a phase of 0. 
In the second step, we undo the spider fusion to the other side. 
Then using \TextCopy and \TextFusion, we can similarly push the $X$ component of the fault to the outputs, resulting in a ZX diagram that can be interpreted as a CNOT followed by a $Y \otimes X$ fault.

Another area where the ZX calculus has proven useful is circuit synthesis, often leading to novel circuit designs, e.g.\@~\parencite{cowtanPhaseGadget2020}.
As a basic example, we synthesize a circuit preparing the 4-qubit cat state.
We translate the cat state to ZX according to \autoref{fig:mappings}, and transform it so that only basic, implementable operations remain:
\[
 \scalebox{.9}{\input{./figures/04-zx/cat-state-prep.tikz}}
\]
As the ZX calculus preserves the underlying linear map, we are thus guaranteed that the resulting circuit exactly implements the  4-qubit cat state.
Note that other rewrites may provide a different implementation.

\section{Faults on ZX Diagrams}
The cat state preparation synthesised above using the ZX calculus is the exact circuit used in \autoref{sec:circuit-fault-equivalence} as a counter-example of how not to implement the cat state in a noisy setting.
While the ZX calculus preserves the underlying linear map, it does not necessarily preserve fault equivalence. 
To reason about fault equivalence in the ZX calculus, we first have to define faults on ZX diagrams. 

\subsection{Faults}

For quantum circuits, we consider Pauli faults in spacetime.
We define spacetime locations as the qubits in between the gates and faults as Paulis acting on the spacetime locations. 
Similarly, for ZX diagrams, we define:

For a circuit $C$ with spacetime locations $\mathcal L$ and a Pauli operator $P \in \overline{\mathcal{P}^{|\mathcal L|}}$, let $C^P$ be the linear map obtained by placing a Pauli $P_j$ at location $j \in \mathcal L$, where $P \propto P_1 \otimes \ldots \otimes P_{|\mathcal L|}$. We call $C^P$ the \textit{application} of $P$ to $C$. Two Paulis are considered equivalent with respect to $C$, written $P \sim Q$ if $C^P = C^Q$, and a Pauli $P$ is called \textit{trivial} with respect to $C$ if $P \sim I$.

\begin{definition}[Spacetime locations, Applying Paulis to a diagram]
 Let $D$ be a ZX diagram with edges $E$. 
 Then the \emph{spacetime locations} of $D$ are the edges $E$. 
 Let $P_{e_1} \otimes \ldots \otimes P_{e_{|E|}} \in \overline{\mathcal{P}^{|E|}}$ be a Pauli operator. 
 Then $D^P$ is the ZX diagram obtained by placing a Pauli $P_e$ each edge $e \in E$. 
 We call $D^P$ the \emph{application} of $P$ to $D$.
\end{definition}

\autoref{fig:fault-examples-ZX} (a), (b) showcase the spacetime locations of an example diagram and an example Pauli operator on that diagram.
\begin{figure}
    \centering
    \begin{minipage}{0.32\textwidth}
        \centering
        \input{./figures/05-ZX-errors/example-diagram-fault-locations.tikz}
        \caption*{(a)}
    \end{minipage}
    \begin{minipage}{0.32\textwidth}
        \centering
        \input{./figures/05-ZX-errors/example-diagram-example-fault.tikz}
        \caption*{(b)}
    \end{minipage}
    \begin{minipage}{0.32\textwidth}
        \centering
        \input{./figures/05-ZX-errors/example-diagram-fault-locations-fault-free.tikz}
        \caption*{(c)}
    \end{minipage}
    \caption{An example ZX diagram from one qubit to two qubits. 
 We have (a) The spacetime locations of the circuit shown as empty octagons. (b) An example fault on the circuit, which, under the edge-flip noise model, is of weight four. (c) The spacetime locations of a submodel of the edge-flip noise model, where three edges are idealised as fault-free (displayed in purple).}
    \label{fig:fault-examples-ZX}
\end{figure}

Once again, we can model faults as Pauli operators acting on some diagram. 
Noise models are, once more, defined as sets of atomic faults that define the effects of individual fault events.
However, while noise models on quantum circuits tend to be dependent on where the qubits flow through the circuit and what the gates are, one key feature of the ZX calculus is that it does not always admit such an interpretation.
For example, when synthesising a quantum circuit, we go via intermediate steps that do not have a physical interpretation. 
To leverage this key property of the ZX calculus, we will define a special noise model that is agnostic to how the ZX diagram corresponds to a quantum circuit: 
\begin{definition}[Edge-flip noise model]
  \label{def:edge-flip-noise}
 Let $D$ be a ZX diagram. The \emph{edge-flip noise model} on $D$ consists of all Paulis of weight one.
\end{definition}
Under edge-flip noise, atomic faults affect exactly one edge, independent of whether this edge is a qubit wire or internal to a gate. 
This means that edge-flip noise is preserved under \TextOCM, i.e.\@ dragging the spiders around.
This rule is essential to reasoning in the ZX calculus and, therefore, desirable to maintain. 
In particular, this preserves the time-agnostic nature of the ZX calculus --- under this noise model, it does not matter how time flows through the diagram. 
This allows us to treat space and time on an equal footing as stipulated by \textcite{gottesmanOpportunitiesChallenges2022}.

Furthermore, we can once again define idealised submodels of edge-flip noise:

\begin{definition}[Submodels of the edge-flip noise model]
  \emph{Submodels of the edge-flip noise model} are subsets of the edge-flip noise model that idealise certain edges as being fault-free by removing the corresponding atomic faults.
\end{definition}

\autoref{fig:fault-examples-ZX} (c) showcases how we might draw an example diagram under some idealised edge-flip noise model, where the idealised edges are drawn in purple.

Throughout the rest of this paper, on ZX diagrams, we will only consider the edge-flip noise models and submodels thereof. 
This will simplify the rewriting of ZX diagrams. 
However, to account for circuit-level noise on quantum circuits, we will have to pay additional attention when translating between quantum circuits and ZX diagrams.

\subsection{Detecting Faults}
\label{sec:detection-regions}

Pauli webs~\parencite{bombinUnifyingFlavorsFault2024} are a decoration of the Clifford ZX calculus tracking how Paulis propagate through a diagram.
They do so by highlighting edges in red and/or green according to a set of simple rules:
\begin{definition}[Pauli web]
 Given a ZX diagram $D$ with edges $E$, a Pauli $P \in \overline{\mathcal{P}^{|E|}}$ defines a highlighting of the edges of $D$ where we consider edge $i$ to be highlighted in green if $P_i = Z$, red if $P_i = X$ and red and green if $P_i = Y$. A Pauli web is an edge highlighting such that:
  \begin{itemize}
    \item a spider with $0$ or $\pi$ phase can have:
    \begin{itemize}
      \item an even number of legs highlighted in its own colour and all or none of its legs highlighted in the opposite colour
    \end{itemize}
    \item a spider with $\pm \frac{\pi}{2}$ phase can have \textbf{either}:
    \begin{itemize}
      \item an even number of legs highlighted in its own colour and no legs highlighted in the opposite colour, \textbf{or}
      \item an odd number of legs highlighted in its own colour and all legs highlighted in the opposite colour
    \end{itemize}
  \end{itemize}
  \label{def:pauli-web}
\end{definition}
\noindent
\autoref{fig:pauli-web-examples} shows an example of some possible Pauli webs on single spiders.

\begin{figure}
  \[
 \input{./figures/04-zx/ZX/pauli-web-examples.tikz}
  \]
  \caption{Example Pauli webs on a single spider.}
  \label{fig:pauli-web-examples}
\end{figure}

Pauli webs are in formal correspondence with the stabilisers and logicals of the diagram~\parencite{borghansZXcalculusQuantumStabilizer2019, rodatzFloquetifyingStabiliser2024}, proving very useful for tracking the flow of information through Clifford diagrams.
However, we will use a different property of Pauli webs; their relationship to detecting sets. 
We define: 
\begin{definition}[Detecting region]
 A \emph{detecting region} is a Pauli web that highlights no input or output edges.
\end{definition}

To understand detecting regions, we observe that Pauli webs locally correspond to stabilisers of the spider up to a global phase. 
In particular, they are either $+1$-stabilisers or $-1$-stabilisers, i.e.\@ they stabilise the spider up to a global factor of $+1$ or $-1$ respectively.
Using this fact, we can now reason: 
\begin{proposition}
  \label{thm:detecting-region}
 Let $D$ be a ZX diagram with a detecting Pauli web $P$.
 Let $n$ be the number of spiders in $D$ where $P$ locally stabilises the spider up to a global phase of $-1$ and $m$ the number of edges highlighted as $Y$.
 Then, if $n + m$ is odd, $D = 0$.
\end{proposition}
\begin{proof}
 Observe that the ZX rules introduced in \autoref{fig:axioms} are up to global scalars --- this is the only place in the paper where we explicitly have to consider global scalars. 
 Pauli webs locally correspond to stabilisers of the spider up to a global phase of $+1$ or $-1$.
 Thus, for each spider, we can surround it with Paulis according to the highlighting of the Pauli web, only changing the global phase of the underlying linear map. 
 For each $-1$-stabiliser, the global phase is multiplied by $-1$.
 Now, as every edge is connected to two spiders and we introduced Paulis of the type of the detecting region surrounding each spider, each edge now has two Paulis of the same type. 
 These cancel each other out.
 Since we have $Y = iXZ$, decomposing a $Y$ error into its $X$ and $Z$ components introduces a global scalar of $i$.
 As each edge highlighted in green and red has two $Y$ spiders, one for each spider that is connected to the corresponding edge, for each of the $m$ edges highlighted as $Y$, we gain a global phase of $i \times i = -1$.
 But then, once we have cancelled out all the newly introduced spiders, we get back the original diagram, up to a $(-1)^{n+m}$ global phase.
 Therefore, we have derived that $D = (-1)^{n+m} D$.
 But then, if $n + m$ is odd, then $D = -D$, implying that $D = 0$.
\end{proof}

So, coming back to the example from \autoref{fig:detecting-sets-examples} (a), two consecutive $ZZ$ Pauli measurements can be expressed as the following ZX diagram with one detecting region. 
We have:
\[\input{./figures/05-ZX-errors/detecting-region-example.tikz}\]
The Pauli web is a $+1$-stabiliser of all the spiders and, thus, for the trivial measurement outcome, the diagram is non-zero. 

When running a non-deterministic circuit with measurements on a quantum computer, we get a measurement outcome sampled from the probability distribution over all outcomes according to the Born rule. 
Any circuit that goes to 0 has a probability of 0 of ever occurring. 
In the example above, this means that any valid measurement outcome must commute with the detecting region, implying that the two projective Pauli measurements must have the same outcome. 

Faults give rise to new ZX diagrams, where we now place additional Pauli spiders on the edges. 
If these occur in a detecting region and they anti-commute with the highlighted edge, they flip the expected parity outcome of the measurements. 
Thus, when running the circuit, we will get the opposite parity as we expected and can conclude that a fault occurred. 
For our example above, if an $X$ fault occurs between the two measurements, we have: 
\[\input{./figures/05-ZX-errors/detecting-region-error.tikz}\]
This means that for the trivial measurement outcome, the diagram is $0$. 
Any measurement outcome that does not make the diagram go to zero must anticommute with the detecting region --- the opposite of what we would expect in the noise-free setting.
More formally, we have:

\begin{definition}[Detectability of Faults]
  \label{def:detectability}
 Let $D$ be a diagram with edges $E$. Let $F \in \overline{\mathcal{P}^{|E|}}$ be a fault on $D$.
 Then $F$ is \emph{detectable} if there exists a detecting region $P \in \overline{\mathcal{P}^{|E|}}$ such that $F$ and $P$ anticommute.
\end{definition}
\begin{proof}
 If $F$ and $P$ anticommute, this means that there are oddly many edges where the Pauli of $F$ is a $-1$-stabiliser of the detecting region. 
 Therefore, for $D$ to be non-zero, the parity of the measurement outcomes in that detecting region has to be the opposite of what it would be in the fault-free case, meaning we get a measurement outcome that violates our expectations --- the fault is detected.
\end{proof}

In contrast to detecting sets that only contain information about correlated measurements, detecting regions additionally contain information about the set of errors detectable by the detecting set. 
Thus, to figure out whether a fault is detectable, in the ZX calculus, one can iterate (a generating set of the) detecting regions and check if the fault anti-commutes with any of them. 
This can be calculated efficiently by adapting the procedure from~\parencite{borghansZXcalculusQuantumStabilizer2019}.

\section{Fault Equivalence for ZX Diagrams}
As outlined in \autoref{sec:circuit-fault-equivalence}, the goal of fault-tolerant circuit synthesis is to derive implementable circuits from a description of a linear map and its behaviour under noise.
One way to achieve this is by adapting the ZX calculus to noisy settings.

In the noise-free setting, circuit synthesis via the ZX calculus can be broken down into three steps:  (1) find a ZX diagram for the desired linear map, (2) rewrite that diagram into a diagram that looks like a quantum circuit, and (3) extract the corresponding circuit.
However, these steps only consider the underlying linear map and do not necessarily preserve behaviour under noise.
Therefore, we have to adapt the procedure.

For steps (1) and (3), where we translate between quantum circuits and ZX diagrams, we observe that we consider different noise models on each; respectively, considering circuit-level noise and edge-flip noise. 
Therefore, we have to take additional care when translating between the two.
For the second step, where we manipulate the ZX diagram, we want to preserve fault equivalence. 
Thus, we have to restrict ourselves to fault-equivalent rewrites. 
Therefore, the overall synthesis process will consist of the following steps:
\begin{enumerate}
  \item Translating a quantum circuit under (a potentially idealised submodel of) circuit-level noise into a fault-equivalent ZX diagram under edge-flip noise.
  \item Rewriting the diagram while preserving behaviour under edge-flip noise.
  \item Extracting a circuit from the rewritten diagram that is fault-equivalent under non-idealised, circuit-level noise.
\end{enumerate}

As all steps preserve fault equivalence, the final circuit will be fault-equivalent to the idealised specification we started off with, resulting in circuits that satisfy our specification by construction.

\subsection{Fault-Equivalent Representation}
We have seen that quantum circuits can be mapped to ZX diagrams; however, these mappings do not consider noise.
We seek a representation that accounts for the behaviour of the quantum circuit under noise. 
We introduce a formal property for ZX diagrams that achieves this strong correspondence:

\begin{definition}[Fault-equivalent representation]
  \label{def:fault-equivalent-repr}
 A ZX diagram $D$ is a \emph{fault-equivalent representation} of a quantum circuit $C$ under some noise model $\mathcal F$ if $D$ under (some submodel of) edge-flip noise is fault-equivalent to $C$ under $\mathcal F$. 
\end{definition}

A fault-equivalent representation of a quantum circuit guarantees that the behaviour of the circuit is fully modelled by the ZX diagram under edge flips.

\subsubsection{Fault Gadgets}
We aim to find ZX diagrams that, under edge-flip noise, are fault-equivalent to quantum circuits under circuit-level noise.
To address this, we introduce fault gadgets—diagrammatic tools for tracking the effect of individual faults.
A fault-equivalent representation can be constructed by mapping the circuit to a fault-free ZX diagram and inserting fault gadgets for each atomic fault.
This immediately guarantees that any fault on the circuit has a natural correspondence in the ZX diagram and vice versa.

First, we define:
\begin{definition}[Pauli boxes]
 We define the following four Pauli boxes:
    \[\input{./figures/06-fault-equivalence-on-ZX/pauli-boxes.tikz}\]
\end{definition}

A fault gadget is a diagrammatic method to track the effect of entangled faults on circuits: 
\begin{definition}
 Let $D$ be a ZX diagram and $P \in \overline{\mathcal{P}^{|E|}}$ be a fault on $D$.
 Then a fault gadget for $P$ can be constructed by adding a fault-free Pauli box on each edge according to $P$, connecting them with fault-free edges to a red spider, which is connected with a normal edge to a green spider.
\end{definition}

Fault gadgets are a method to track larger-weight faults while preserving the assumption that faults are represented by independent edge flips. 
They do not correspond to a physical operation within a circuit but rather are a diagrammatic tracker of potential faults.

For example, for a diagram consisting of the identity on four qubits and the fault $I \otimes X \otimes Y \otimes Z$, we have: 
\[\input{./figures/06-fault-equivalence-on-ZX/fault-gadget.tikz}\]
More generally, fault gadgets clearly generalise to arbitrary weight Paulis. 
By spider fusion, Pauli boxes corresponding to identity edges can be fused in and, therefore, omitted. 

In a noise-free setting, fault gadgets correspond to the identity:
   \[\input{./figures/06-fault-equivalence-on-ZX/fault-gadget-semantics.tikz}\]
This means that adding a fault gadget does not change the linear map a diagram represents. 

However, it changes its behaviour under noise:
   \[\input{./figures/06-fault-equivalence-on-ZX/fault-gadget-noise.tikz}\]
A single $Z$ flip on the fault gadget creates a fault that is equivalent to an $I \otimes X \otimes Y \otimes Z$ fault on the output. 
Thus, by adding the fault gadgets, we have introduced an additional spacetime location to the ZX diagram. 
A single $Z$ edge flip in that location has the effect of the fault we wanted to track. 
Furthermore, as the single-legged green spider stabilises $X$, $X$ flips on that edge are trivial, and $Y$ flips are equivalent to $Z$ flips. 
Therefore, no other non-trivial faults are added by the fault gadget.

\begin{proposition}
 For any quantum circuit $C$ under any noise model $\mathcal F$, we can find a ZX diagram that is fault-equivalent under edge-flip noise. 
\end{proposition}
\begin{proof}
 First, we translate the circuit using the mapping in \autoref{fig:mappings}, idealising all edges as fault-free.
 Then, for each atomic fault in the circuit, we insert the corresponding fault gadget.
 As a result, the only possible faults in the ZX diagram are those represented by these gadgets, directly mirroring the circuit's faults.
 Therefore, the final ZX diagram is fault-equivalent to the original circuit.
\end{proof}

\subsubsection{Fault-Equivalent Representations with Fewer Fault Gadgets}
\label{sec:phenomenological-fe-representation}
While fault gadgets are a useful tool for capturing highly entangled noise under edge flips, they can be tedious to work with. 
As circuits get larger and the noise models more involved, the number of necessary fault gadgets becomes unwieldy. 
We observe that often, more basic diagrams without fault gadgets are also fault-equivalent for specific circuits under certain idealised submodels of circuit-level noise.

There are two translation processes between quantum circuits and ZX diagrams, once when translating the specification into a ZX diagram and once when translating the final ZX diagram back into a circuit. 
First, we will consider the specifications. 

\begin{proposition}
  \label{prop:fault-free}
 For any gate idealised as fault-free, we can find a fault-equivalent ZX diagram without using fault gadgets. 
\end{proposition}
\begin{proof}
 As the gate is assumed to be fault-free, the only faults that can occur on the quantum circuit are qubit flips on the input and output edges. 
 But then, we can translate the gate according to \autoref{fig:mappings}, idealising all internal edges of the ZX diagram as fault-free. 
 Now, the only faults allowed on the corresponding ZX diagram are edge flips on the input and output edges. 
 These correspond exactly to the qubit flips.
 Therefore, the diagram will be fault-equivalent to the fault-free gate. 
\end{proof}

So, for example, we have the following fault equivalences: 
\[\input{./figures/06-fault-equivalence-on-ZX/fault-free-equivalence-examples.tikz}\]

\begin{corollary}
  \label{cor:fault-equivalent-unitaries}
 Single-qubit unitaries can be represented by fault-equivalent ZX diagrams without using fault gadgets.
\end{corollary}
\begin{proof}
 Idealising gates as fault-free removes atomic faults that act on subsets of their output qubits. 
 As single-qubit gates only have one output qubit, their faults are already accounted for by the corresponding qubit faults on their output wire. 
 But then, single-qubit unitaries are fault-equivalent to their idealised self. 
 Thus, by \autoref{prop:fault-free}, they can be represented by ZX diagrams without requiring fault gadgets. 
\end{proof}

For multi-qubit Pauli measurements, we proposed two levels of idealisation, fault-tolerant measurements and the stronger idealisation of fault-free measurements.
By \autoref{prop:fault-free}, we can easily represent fault-free Pauli measurements. 
However, for fault-tolerant measurements, we additionally allow for measurement flips, modelled by an anticommuting Pauli before and after the measurement, and measurement flips along with single-qubit flips on the outputs.
We have:

\begin{proposition}
  \label{prop:fault-equivalent-ft-pauli}
 For fault-tolerant multi-qubit Pauli-$Z$ measurements, we have:
    \[\input{./figures/06-fault-equivalence-on-ZX/fault-tolerant-measurement.tikz}\]
\end{proposition}
\begin{proof}
 To show the fault equivalence, we will argue that all atomic faults on either side have natural atomic correspondences on the other side. 

 The measurement flip corresponds to $X$ flips on one of the internal edges of the measurement:
    \[\input{./figures/06-fault-equivalence-on-ZX/measurement-flip.tikz}\]
 A measurement flip, along with a $Z$ data flip on one of the output qubits, corresponds to a $Y$ flip on the corresponding measurement edge, pushing the measurement flip through the measurement spider accordingly:
    \[\input{./figures/06-fault-equivalence-on-ZX/measurement-data-flip.tikz}\]
 Measurement flips, along with $X$ or $Y$ data flips on one of the outputs, are equivalent to a corresponding $X$ or $Y$ flip on the input, as pushing this fault from the input to the outputs creates the output flip and the measurement flip. 
 But then all atomic faults on fault-tolerant multi-qubit Pauli-$Z$ measurements have corresponding atomic edge flips on the corresponding ZX diagram. 
 Thus, any fault on the idealised circuit, which is some combination of atomic faults, must naturally have a corresponding fault of at most the same weight on the ZX diagram. 
    
 Conversely, all edge flips correspond to some atomic fault of the measurement. 
 Edge flips on the inputs and outputs naturally correspond to qubit flips. 
 An $X$ flip on the internal edges corresponds to a measurement flip (as shown above). 
 A $Z$ flip can be pushed out to become a $Z$ flip on the corresponding output edge
 , and a $Y$ flip corresponds to a measurement flip along with a $Z$ flip (as shown above). 

 But then, we have shown that the two diagrams are fault-equivalent. 
\end{proof}

For all other multi-qubit measurements, we observe that they can be created by conjugating the Pauli-$Z$ measurement with single-qubit Cliffords. 
As single-qubit unitaries have natural fault-equivalent diagrams (\autoref{cor:fault-equivalent-unitaries}) and the composition of fault-equivalent diagrams preserves fault equivalence (\autoref{prop:fault-equivalence-compositionality}), we can represent any multi-qubit Pauli measurement fault-equivalently. 

Beyond translating idealised specifications into ZX diagrams, we also have to extract quantum circuits from our ZX diagrams. 
However, for that, we have to consider entangled faults that act on any of the outputs of the gate. 
We have: 

\begin{proposition}
  \label{prop:CNOT-fault-accounting}
 The CNOT under circuit-level noise is fault-equivalent to the following ZX diagram:
  \[
 \input{./figures/04-zx/fault-accounting/CNOT.tikz}
 \quad \FaultEq \quad
 \input{./figures/04-zx/fault-accounting/fault-accounting-cnot.tikz}
  \]
\end{proposition}
\begin{proof}
 Under circuit-level noise, a fault event on the CNOT can create any non-trivial fault in $\overline{\mathcal{P}^{2}}$ on the output qubits; there are 15 potential faults.
  
 We observe that faults that only affect one output qubit are naturally accounted for by a corresponding edge flip. 
 Additionally, to these six faults, the following five faults are also accounted for by single edge flips:
  \[\input{./figures/04-zx/fault-accounting/cnot-fault-accounted.tikz}\]

 But then, we have accounted for 11 out of the 15 potential faults. 
 The remaining four faults are accounted for by the corresponding fault gadgets. 
\end{proof}

As such, we can not freely translate a CNOT-like ZX diagram under edge-flip noise into a CNOT quantum circuit under circuit-level noise, as the two are not fault-equivalent. 
Compared to the 15 possible faults on the circuit, only requiring to account for four fault gadgets to make the translation fault-equivalent is already a substantial improvement.
However, going a step further, there are composite circuits that include CNOTs that naturally account for all faults of circuit-level noise without requiring fault gadgets:

\begin{proposition}
  \label{prop:CNOT-plus-measurement}
 Let $C$ be the circuit consisting of a CNOT followed by a Z-basis measurement on the second qubit. 
 Then: 
  \[
 \input{./figures/04-zx/fault-accounting/cnot-measurement.tikz}
 \quad \FaultEq \quad
 \input{./figures/04-zx/fault-accounting/fault-accounting-cnot-measurement-simplified.tikz}
  \]
\end{proposition}
\begin{proof}
 Due to the fact that the measurement stabilises $Z$, for each fault previously unaccounted for, there is now a corresponding single edge flip that produces the required effect:
  \[
 \input{./figures/04-zx/fault-accounting/cnot-measurement-proof.tikz}
  \]
\end{proof}

Similarly, we can state: 
\begin{proposition}
  \label{prop:CNOT-plus-state}
 Let $C$ be the circuit consisting of preparing the first qubit in $\ket{+}$ followed by a CNOT gate controlled by the first qubit.
 Then:
  \[
 \input{./figures/04-zx/fault-accounting/cnot-plus-state.tikz}
 \quad \FaultEq \quad
 \input{./figures/04-zx/fault-accounting/fault-accounting-cnot-plus-state.tikz}
  \]
\end{proposition}
\begin{proof}
 Analogously to the above, we can iterate all missing faults on the CNOT and show that they are now accounted for by single edge flips.
\end{proof}

These two composite circuits will become very useful when we construct fault-tolerant circuits.
While they might seem ad hoc, we can view them as fault-equivalent implementations of the split and merge operations, as, for example, known from lattice surgery \parencite{debeaudrapZXCalculus2020}.
They are easy components to identify as fault-equivalent under circuit-level noise. 

By \autoref{prop:fault-equivalence-compositionality}, we know that if we can break a larger circuit into fault-equivalent subcomponents, the composite circuit is also fault-equivalent. 
Thus, if we can break a ZX diagram into the components above, as well as single-qubit operations, we know that the corresponding quantum circuit is fault-equivalent. 

\subsection{Fault-Equivalent Rewrites}
\label{subsec:fe-rewrites}
Having translated a circuit description to a ZX diagram, the second step of our proposed circuit synthesis framework is to manipulate the ZX diagram to bring it into the shape of a circuit consisting of implementable gates. 
However, as we want to preserve fault equivalence, we can not simply use any ZX rewrites; we have to restrict ourselves to \emph{fault-equivalent rewrites}, originally introduced by \textcite{rodatzFloquetifyingStabiliser2024}:

\begin{definition}[$w$-fault-equivalent rewrite, fault-equivalent rewrite]
  \label{def:fault-tolerant-rewrites}
 A ZX rewrite $r: D_1 \to D_2$ is a \textit{$w$-fault-equivalent rewrite} if it preserves $w$-fault equivalence of $D_1$ and $D_2$ under edge-flip noise.
 It is \textit{fault-equivalent} if it preserves fault equivalence.
\end{definition}

An example of a fault-equivalent rewrite is the following:
\begin{proposition}
  \label{thm:elim-rewrite}
 The following rewrite is fault-equivalent:
  \begin{gather}
 \tag{$\textsc{Elim}_{\text{fe}}$}\label{elim-rewrite}\refstepcounter{equation}
 \input{./figures/04-zx/fault-equiv-rewrites/elim-rewrite.tikz}
  \end{gather}
\end{proposition}
\begin{proof}
 Any edge flip on the left diagram corresponds directly to an edge flip of the same type on either of the legs of the right diagram.
 Similarly, any combination of edge flips on the right diagram of types $P_1, P_2$ corresponds to an edge flip of type $P_1 P_2$ on the identity wire.
\end{proof}

A slightly more interesting fault-equivalent rewrite for spiders with at least one output leg is the following:
\begin{proposition}
  \label{thm:fuse-rewrite}
 The following rewrite is fault-equivalent:
  \begin{gather}
 \tag{$\textsc{Fuse-1}_{\text{fe}}$}\label{fuse-rewrite}\refstepcounter{equation}
 \input{./figures/04-zx/fault-equiv-rewrites/fuse-rewrite.tikz}
  \end{gather}
\end{proposition}
\begin{proof}
 We observe that a sufficient condition for fault equivalence is that faults on internal edges can be pushed to boundary edges without increasing their weight. 
 As the boundary edges have natural correspondences between the two diagrams, if that condition is satisfied, we have found an equivalent fault on the other diagram with at most the same weight.
 This idea is more formally stated and proven in \autoref{lem:fault-equiv-boundaries}.

 As only one of the diagrams has an internal edge, there are three faults that could be problematic: an $X, Y$ or $Z$ flip of that internal edge. 
 However, we observe: 
  \[\input{./figures/04-zx/fault-equiv-rewrites/fuse-rewrite-proof.tikz}\]
 As these faults are either trivial or can be pushed out to the boundary edges without increasing in weight, they have a natural correspondence on the other side. 
 Thus, we have shown that all faults on either side have an equivalent fault of at most the same weight on the other side, and, thus, the rewrite is fault-equivalent.
\end{proof}

The fault equivalence of the above rewrites is due to the correspondence in faults, meaning that any fault on the left has a corresponding fault of at most the same weight on the right and vice versa. 
However, for faults to satisfy the fault equivalence relationship, alternatively, they are acceptable if they are detected. 
Thus, an alternative way to achieve fault equivalence is by introducing detecting regions to detect problematic faults.

\begin{proposition}
  \label{thm:four-legged-spider}
 The following rewrite is fault-equivalent:
  \begin{gather}
 \tag{$\textsc{Fuse-4}_{\text{fe}}$}\label{four-legged-rewrite}\refstepcounter{equation}
 \input{./figures/04-zx/fault-equiv-rewrites/four-legged.tikz}
  \end{gather}
\end{proposition}
\begin{proof}
 The four Z-spiders form a detecting region, and therefore, any odd-weight $X$ fault in the internal edges is detectable.  
  \[
 \input{./figures/04-zx/fault-equiv-rewrites/four-legged-detecting-region.tikz}
  \]
 Any $X$ fault with an even weight can be pushed out without increasing the weight of the fault:
  \[
 \input{./figures/04-zx/fault-equiv-rewrites/four-legged-X-flips.tikz}
  \]
 As the diagram only consists of Z-spiders, all $Z$ faults can be pushed out without increasing the weight of the fault.
 Similarly, $Y$ faults can be decomposed into their $X$ and $Z$ components and pushed out correspondingly.
 But then, by \autoref{lem:fault-equiv-boundaries}, the two diagrams are fault-equivalent.
\end{proof}

Next, we mention a family of fault-equivalent rewrites whose preservation of fault equivalence is slightly more complicated to prove:
\begin{restatable}{proposition}{inductiveDistRewrite}
  \label{thm:inductive-plus-dist-rewrite}
 The following rewrite for $2n$-legged spiders is fault-equivalent:
  \begin{gather}
 \tag{$\textsc{Fuse-n}_{\text{fe}}$}\label{n-legged-rewrite}\refstepcounter{equation}
 \input{./figures/04-zx/fault-equiv-rewrites/recursive-spider.tikz}
  \end{gather}
\end{restatable}
\begin{proof}
 See \autoref{subsec:fault-equivalent-rewrite-proofs}.
\end{proof}

So far, all rewrites we presented preserve fault equivalence, i.e.\@ $w$-fault equivalence for a $w \in \mathbb N$.
However, in some contexts, it is sufficient and potentially cheaper to preserve $w$-fault equivalence. 
Therefore, we present one more rewrite:
\begin{restatable}{proposition}{inductiveDistRewriteW}
  \label{thm:inductive-plus-dist-rewrite-w}
 The following rewrite for $2n$-legged spiders is $w$-fault-equivalent:
  \begin{gather}
 \tag{$\textsc{Fuse-n}_{w\text{-fe}}$}\label{n-legged-rewrite-w}\refstepcounter{equation}
 \input{./figures/04-zx/fault-equiv-rewrites/recursive-spider-w.tikz}
  \end{gather}
\end{restatable}
\begin{proof}
 See \autoref{subsec:fault-equivalent-rewrite-proofs}.
\end{proof}

\section{Examples: Fault-Tolerant Circuit Synthesis and Optimisation}
\label{sec:examples}

This section applies our framework to demonstrate its utility in fault-tolerant circuit design.
In particular, we follow the methodology outlined in the previous section to find fault-tolerant implementations for cat state preparation and syndrome extraction circuits; (1) we first translate an idealized quantum circuit into a fault-equivalent ZX diagram, (2) then we rewrite the diagram using fault-equivalent rewrites so that only elements remain that are fault-equivalent to implementable components of a quantum circuit, (3) enabling circuit extraction.
As each step involved in this method preserves fault equivalence, we obtain a fault-equivalent implementation of the given specification and thus, depending on the starting specification, circuits that fault-tolerantly implement the desired gadgets.

This section first discusses the fault-tolerant cat state preparation.
We then proceed to study Shor~\parencite{shorFaulttolerantQuantum1996} and Steane-style~\parencite{steaneActiveStabilization1997} syndrome measurements and extractions.
Our framework can not only verify these methods but also optimise them and provide alternatives.

\subsection{Fault-Tolerant Cat State Preparation}

Fault-tolerant cat state preparation (up to some weight $w$) is the task of finding a circuit that prepares the cat state such that undetectable faults (of weight less than $w$) do not propagate out to create more faults on the data qubits. As previously defined, we have:

\ftCatStatePrep*

Before providing a family of circuits that prepare arbitrary-legged cat states, we first look at the cat state on four qubits:
\begin{proposition}
  \label{prop:four-legged-cat}
 We can fault-tolerantly prepare a cat state on four qubits as follows:
  \[\input{./figures/07-examples/cat-state/four-legged-statement.tikz}\]
\end{proposition}
\begin{proof}
 We have:
  \[
 \input{./figures/07-examples/cat-state/four-legged-proof.tikz}
  \]
 The first step of our synthesis procedure is to translate the idealised description of the desired map into a fault-equivalent ZX diagram.
 By \autoref{prop:fault-free}, the map from \autoref{fig:mappings} is already fault-equivalent, as any edge flip on the ZX diagram exactly corresponds to a flip on the data qubits.
 Then, once considering ZX diagrams under edge-flip noise, in the second step, we can use fault-equivalent rewrites to manipulate the diagram.
 Using a sequence of rewrites from the previous section, we get a new ZX diagram that looks like a quantum circuit.
 To extract the final circuit, we can observe that the circuit is built out of the component from \autoref{prop:CNOT-plus-measurement} and \autoref{prop:CNOT-plus-state}.
 We mark the respective fault-equivalent components in the ZX diagram and the quantum circuit using dotted boxes.
 Therefore, by the compositionality of fault equivalence (\autoref{prop:fault-equivalence-compositionality}), we know that the ZX diagram is fault-equivalent to the final circuit.
\end{proof}

Thus, we have found a (well-known) implementation of the cat state on four qubits.
We can read the circuit as instantiating a flag qubit which gets entangled with the first qubit before preparing the cat state using the non-fault-tolerant method, previously described in \autoref{sec:zx-rewrites}.
By then entangling the flag with the last qubit before measuring it out, the flag catches the faults that would otherwise propagate badly.

The presented fault-tolerant circuit synthesis protocol is fundamentally different from other protocols in the literature.
While other flag-based approaches~\parencite{chamberlandFlagFaulttolerant2018, pehamAutomatedSynthesis2025} iterate all possible faults up to some weight, identify detrimental ones, and construct corresponding flags to catch them, we have synthesised a circuit that is correct by construction.
The necessary flag to make the circuit fault-tolerant is a consequence of our construction. 
We are guaranteed that any fault on the final circuit must either be detectable or correspond to a fault of at most the same weight on the idealised cat state.
As we only used fault equivalences, i.e.\@ no $w$-fault equivalences, we are guaranteed that this must hold for any fault, no matter the weight.

Using this fault-tolerant circuit synthesis method, we can now construct an arbitrary legged cat state at arbitrary distances:
\begin{proposition}
  \label{prop:arbitrary-legged-cat}
 We can construct a fault-tolerant $2n$-legged cat state from two fault-tolerant $n$-legged cat states by transversally, fault-tolerantly measuring $n$ $ZZ$ parity checks between the two $n$-legged cat states:
  \[
 \input{./figures/07-examples/cat-state/arbitrary-legged-statement.tikz}
  \]
\end{proposition}
\begin{proof}
 We have:
  \[\resizebox{\textwidth}{!}{\input{./figures/07-examples/cat-state/arbitrary-legged-proof.tikz}}\]
\end{proof}

While the proposed procedure only works for cat states on $n = 2^k$ qubits, using~\eqref{fuse-rewrite}, it can be adapted to arbitrary-legged cat states.

We have presented a family of circuits that fault-tolerantly prepare arbitrary legged cat states at arbitrary distances, meaning that any fault, no matter the weight, is either detectable or creates at most as many data errors.
However, while the circuits above prepare the cat state fault-equivalently, in many contexts we only care about fault equivalence up to some weight, i.e.\@ $w$-fault equivalence.
This enables a cat state preparation with significantly fewer gates, as for example shown by  \textcite{prabhuFaulttolerantSyndrome2023} at distance three.
Using \autoref{thm:inductive-plus-dist-rewrite-w}, we have:

\begin{proposition}
  \label{prop:arbitrary-legged-cat-w}
 We can construct a $w$-fault-tolerant $2n$-legged cat state from two fault-tolerant $n$-legged cat states by transversally, fault-tolerantly measuring $w$ $ZZ$ parity checks between the first $w$ legs of the two $n$-legged cat states:  \[
 \input{./figures/07-examples/cat-state/arbitary-legged-w-statement.tikz}
  \]
\end{proposition}
\begin{proof}
 We have:
  \[\input{./figures/07-examples/cat-state/arbitary-legged-w-proof.tikz}\]
 for $k_1 = k_2 = \dots = k_{w} = k$.
\end{proof}

This new protocol only requires $2w$ CNOTs and $w$ measurements instead of $2n$ CNOTs and $n$ measurements, potentially leading to significantly smaller circuits.

Thus, we have used fault equivalence and fault-equivalent rewrites in the ZX calculus to recover existing (\autoref{prop:four-legged-cat}) and find new (\autoref{prop:arbitrary-legged-cat}, \autoref{prop:arbitrary-legged-cat-w}) circuits for cat state preparation whose fault equivalence is satisfied by construction.

\subsection{Fault-Tolerant Syndrome Extraction}
For a given QEC code with distance $d$, fault-tolerant syndrome extraction is about perfectly measuring all the stabilisers of that code without any faults occurring in between:
\ftSyndromeExtraction*

In this section, we will verify the correctness of two protocols that solve this task; Shor-style syndrome extraction~\parencite{shorFaulttolerantQuantum1996} and Steane-style syndrome extraction~\parencite{steaneActiveStabilization1997},
and show how these protocols can be optimised and altered.

\subsubsection{Shor-style Syndrome Extraction}\label{sec:shor-example}
\textcite{shorFaulttolerantQuantum1996} separates the problem of syndrome extraction into two steps, first, he gets rid of the idealised edges and measurements by repeating the measurement schedule, then he implements the individual syndrome measurements fault-tolerantly.

The first step of Shor-style syndrome extraction builds on the following observation:
\begin{proposition}
  \label{prop:repeating-measurement}
 Let $S = [S_1, \dots, S_m]$ be the stabilisers of some stabiliser code with distance $d$.
 Then, repeating the measurement schedule with imperfect, yet fault-tolerant stabiliser measurement $d$ many times, $d$-fault-equivalently implements the syndrome extraction:
  \[
 \input{./figures/07-examples/shor/claim-repeat-stabilisers.tikz}
  \]
\end{proposition}
\begin{proof}
 The proof proceeds by showing that all non-detectable internal faults (weight $<d$) can be pushed to the boundaries without increasing their weight.

 As we idealised the measurements as fault-tolerant, each atomic fault can create at most one qubit flip and/or a measurement flip.
 Since there are $d$ iterations of the measurement schedule, a fault of weight less than $d$ must leave at least one iteration fault-free.
 Consider all faults to the left of this fault-free iteration.
 We can push all faults that have a qubit flip component to the input wires.
 For all the measurements they anticommute with, this introduces a measurement flip of that measurement.
 Eventually, all qubit flips are pushed to the inputs, possibly creating additional flips of some measurements.
 We can analogously push all qubit flips to the right of the perfect round to the outputs.

 At this point, we have found an equivalent fault that consists of qubit flips on the inputs and outputs, as well as measurement flips.
 While the total number of atomic faults may have increased, we observe that the number of qubit flips does not increase when pushing the qubit flips past measurements.
 Therefore, the number of qubit flips is at most the weight of the original fault.
 The only potential faults remaining within the $d$ repetitions are measurement flips.
 However, since one round of measurements is fault-free, and we assumed the fault to be undetectable, all other iterations of the measurements must have the same outcomes as the fault-free one (otherwise, differing repeated measurements of the same stabiliser would give different values, making the fault detectable).
 Therefore, all measurement flips must cancel each other out.
 But then, the fault we got from pushing out all the edge flips only acts on the inputs and outputs and, in particular, is at most of the same weight as the original fault.
\end{proof}
When we express this proposition using ZX diagrams, we can think of it as a method for removing idealised, fault-free edges.
In the original circuit, the fault-equivalent ZX diagram assumes that all internal edges are idealised.
However, by \autoref{prop:fault-equivalent-ft-pauli}, the repeated circuit on the right can be represented as a fault-equivalent ZX diagram that contains no fault-free edges.
This is because fault-tolerant measurements can be expressed as fault-equivalent ZX diagrams without idealised edges.

The second step of Shor-style syndrome extraction is to fault-tolerantly implement the individual, non-perfect syndrome measurements.

\paragraph{Verifying Shor-style syndrome measurements}
Verifying the correctness of existing protocols boils down to checking whether proposed circuits are fault-equivalent to the idealised specification. 
In practice, it is easier to start with the complicated protocol and try to reduce it to the often simpler specification. 
However, in this case, we will start with the specification and show that we can rewrite it to Shor-style syndrome extraction. 

Shor-style syndrome extraction uses the $H\ket{cat}$ state, i.e.\@ $\ket{cat}$-state in the $X$ basis $\ket{++++} + \ket{----}$. 
To measure an $n$ qubit $Z$ stabiliser, Shor-style syndrome extraction fault-tolerantly prepares a $H\ket{cat}$ state on ancillary qubits, entangles it with the data qubits, and measures the ancillary qubits.

Assuming the fault-tolerant preparation of the $H\ket{cat}$ state, our framework can be used to verify the correctness of this procedure as follows:
\begin{proposition}
  \label{prop:shor-ft}
 Shor-style syndrome measurements are fault-tolerant:
  \[
 \input{./figures/07-examples/shor/shor-ft-statement.tikz}
  \]
\end{proposition}
\begin{proof}
 We have:
  \[
 \input{./figures/07-examples/shor/shor-ft-proof.tikz}
  \]
\end{proof}

While the above circuit only performs Pauli-$Z$ measurements, we can implement arbitrary Pauli measurements by conjugating the measurement using single-qubit Cliffords.
Thus, using \autoref{prop:repeating-measurement} and \autoref{prop:shor-ft}, we can do the syndrome extraction by repeating Shor-style syndrome measurements $d$ many times.

We observe that Shor-style syndrome extraction uses four measurements to implement one stabiliser measurement. 
In this case, we can derive the outcome of the stabiliser measurement by taking the parity of the four destructive measurements.
In \autoref{appendix:measurement-outcomes}, we elaborate more generally on how to translate between measurement outcomes of equivalent Clifford circuits.

\paragraph{Optimising Shor-style syndrome measurements}

An interesting aspect of our fault tolerance proof for the Shor-style syndrome measurement is that it does not rely on transversality, the usual argument for its correctness.
This implies that we can also verify non-transversal constructions.
In particular, we can take Shor-style syndrome extraction and optimise it:

\begin{proposition}
  \label{prop:shor-optimised}
 We can optimise Shor-style syndrome measurements as follows:
  \[\input{./figures/07-examples/shor/shor-optimised-complete-statement.tikz}\]
\end{proposition}
\begin{proof}
 We have:
  \[\input{./figures/07-examples/shor/shor-optimised-proof.tikz}\]
\end{proof}

For Shor-style syndrome extraction, we have to fault-tolerantly prepare the $H\ket{cat}$ state. 
For this optimised procedure, we have to fault-tolerantly implement an $H\ket{cat}$-like linear map:
\[\input{./figures/07-examples/shor/cat-state-vs-cat-like-map.tikz}\]

We can observe that to fault-tolerantly prepare the $H\ket{cat}$, we require one auxiliary qubit and five CNOTs~(\autoref{prop:four-legged-cat}).
Here, we propose a fault-tolerant implementation of the $H\ket{cat}$-like linear map that only requires four CNOTs and one auxiliary qubit. 
Thus, overall, the optimised method only requires three auxiliary qubits, which is two fewer than Shor's original syndrome extraction. 
More generally, we can apply this optimisation to arbitrary weight Shor-style syndrome measurements to save on the qubit count. 

The main trade-off of this new method is that the $H\ket{cat}$ state in the original Shor method can be prepared asynchronously; if detectable faults occur during preparation, the state can be discarded, and a new state can be prepared.
In our construction, the $H\ket{cat}$-like linear map gets entangled with the data qubits during its implementation, so faulty implementations cannot be discarded and must instead be handled by the outer code.
However, if the original diagram had some distance $d$, our use of fault-equivalent rewrites preserves circuit distance by \autoref{cor:fe-preserves-dist}, meaning the resulting diagram still has distance $d$.
Thus, by \autoref{prop:fault-equivalence-compositionality}, if the overall number of faults (including those within the $H\ket{cat}$-like map) is at most $\frac{d-1}{2}$, they remain correctable.

\paragraph{Alternative Shor-like syndrome measurements}
Beyond optimising Shor-style syndrome measurements, we can also come up with completely new ones.
Inspired by \textcite{rodatzFloquetifyingStabiliser2024}, we have:
\begin{proposition}
  \label{prop:shor-alternative}
 The following circuit fault-tolerantly performs a syndrome measurement:
  \[\input{./figures/07-examples/shor/shor-rodatzetal-statement.tikz}\]
\end{proposition}
\begin{proof}
 We have:
  \[
 \input{./figures/07-examples/shor/shor-rodatzetal-proof.tikz}
  \]
\end{proof}

This alternative implementation of a syndrome measurement is equally efficient as our optimisation of Shor-style syndrome measurements.
Additionally, it moves data qubits around the chip in a teleportation-like fashion; the two auxiliary qubits end up carrying data while two of the data qubits are freed up.
This can lead to dynamic implementations of syndrome measurement circuits akin to \textcite{eickbuschDemonstratingDynamic2024}, which can aid with routing on chips with limited connectivity or with leakage reduction.

\subsubsection{Steane-style Syndrome Extraction}
An alternative method for fault-tolerant syndrome extraction is proposed by \textcite{steaneActiveStabilization1997}.
Once again, the underlying question is to fault-tolerantly implement the specification of perfectly measuring all stabilisers without intermediate faults.

\paragraph{Verifying Steane-style syndrome extraction}
We focus on Steane-style syndrome extraction for CSS codes, where we can assume that the stabilisers are separable into  $X$- and $Z$-stabilisers.
We have:
\begin{proposition}
  \label{prop:steane-correctness}
 Let $S = [S_1, \dots, S_m]$ be the stabilisers of some CSS stabiliser code with distance $d$.
 Then:
  \[
 \input{./figures/07-examples/steane/steane-statement.tikz}
  \]
\end{proposition}
\begin{proof}
 As $S$ is a CSS code, we can separate the stabilisers into $X$ stabilisers $[S_{1}, \dots, S_{i}]$ and $Z$ stabilisers $[S_{i+1}, \dots, S_{{m}}]$, we get:
  \[\input{./figures/07-examples/steane/steane-proof.tikz}\]
 where we make use of the fact that $X$ stabilisers are the colour inverse of $Z$ stabilisers.
 The exact connectivity of the stabilisers is dependent on the specific choice of code.
 To account for different codes more formally, Steane-style syndrome extraction can be treated using the scalable ZX-calculus~\parencite{caretteSZXCalculusScalable2019}, as illustrated by \textcite{KissingerWetering2024Book}, however, without considering fault equivalence.
 To avoid introducing more notation, we instead depict this by randomly omitting some of the edges for each stabiliser.

 The first step of the correctness proof uses \autoref{lem:fe-perfect-fuse} to fuse the spiders on the data qubits.
 Then, using \autoref{lem:fe-perfect-fuse} a second time, we unfuse the merged spiders.
 The next step is $d$-fault-equivalent only if, in the fault-free case, the input state lives in the code space.
 In particular, we observe that if the input state is in the code space, then any fault occurring on the newly introduced non-idealised edges of weight less than $d$ can either be pushed to the outputs or must be detectable, detecting problematic faults that would otherwise spread unfavourably.
 For more details, see \autoref{appendix:steane-proof}.
 After having removed some of the idealised edges, similar to Shor-style syndrome extraction, we can use \autoref{lem:fe-pi-copy}, \eqref{fuse-rewrite} and \TextOCM to get the final circuit.
 Thus, we have proven the correctness of Steane-style syndrome extraction.
\end{proof}
Here, recovering this correctness proof through the lens of fault equivalence allows for an interesting insight: Steane-style syndrome extraction does not eliminate idealised edges through repetition, like Shor's method.
Instead, it transforms the problem: fault-tolerantly implementing the initial auxiliary logical states becomes the central challenge.

\paragraph{Optimising Steane-style syndrome extraction}
Similar to how we optimised Shor-style syndrome measurements, we can optimise Steane-style syndrome extraction by bending some of the legs around.
We have:
\begin{proposition}
  \label{prop:optimised-steane}
 We can optimise Steane-style syndrome extraction as follows:
  \[\input{./figures/07-examples/steane/steane-optimisation-statement.tikz}\]
 where
  \[\input{./figures/07-examples/steane/def-logical-state-like.tikz}\]
\end{proposition}
\begin{proof}
 First, we observe:
  \begin{align}
    \label{eq:steane-helper}
 \scalebox{0.9}{\input{./figures/07-examples/steane/state-like-proof.tikz}}
  \end{align}
 But then we have:
  \[\input{./figures/07-examples/steane/steane-syndrome-optimisation.tikz}\]
\end{proof}

Thus, we have bent around half the legs such that we only need half as many auxiliary qubits and measurements.
Similar to how the optimised Shor-style syndrome extraction required a fault-tolerant implementation of the $H\ket{cat}$-like linear map, we now require a fault-tolerant implementation of the $\ket{0}$-like linear map.
Only if this linear map is sufficiently efficient to implement does the overall optimisation reduce overhead.
While we do not provide a general procedure to implement $\ket{0}$-like linear maps, we will consider the Steane code~\parencite{steaneMultipleparticleInterference1997} as one example where we can implement the $\ket{0}$-like linear map efficiently.

\paragraph{Example: Steane-style syndrome extraction for the Steane code}
Efficiently preparing a logical state, and similarly efficiently executing a $\ket{0}$-like linear map, are hard problems.
Therefore, instead of starting from scratch, to implement the $\ket{0}$-like linear map for the Steane code, we will take an efficient logical state preparation for the Steane code and adapt it for our purposes.
By \textcite{gotoMinimizingResource2016}, we know that for the Steane code, we have:
\begin{proposition}
 For the Steane code, the logical $\ket{0}$ state can be $3$-fault-tolerantly prepared as follows:
  \[
 \input{./figures/07-examples/steane/steane-code-logical-state.tikz}
  \]
\end{proposition}
The circuit is only fault-equivalent up to weight three.
However, as the Steane code only has distance three, larger fault equivalences do not add any value.
Additionally, the final ZX diagram is fault-equivalent to the extracted circuit only under circuit-level CSS noise, i.e.\@ considering circuit-level $X$ faults and $Z$ faults separately.
For CSS codes, circuit-level CSS noise is a sufficient approximation of circuit-level noise.
Under circuit-level CSS noise, in the ZX calculus, CNOT gates without any idealised edges are naturally fault-equivalent to CNOT gates.
We note that CSS noise does not neatly fit into our current definition of a noise model, as our noise model considers the entire group generated by the atomic faults.

Using this efficient implementation of the logical state, we can now derive an efficient implementation of the $\ket{0}$-like linear map.
As it does not matter which edges we bend around, we pick ones that are particularly easy for us:
\begin{proposition}
  \label{prop:optimised-0-like}
 For the Steane code, the $\ket{0}$-like linear map can be $3$-fault-tolerantly prepared as follows:
  \[\input{./figures/07-examples/steane/state-like-derivation-statement.tikz}\]
\end{proposition}
\begin{proof}
  \[\resizebox{\textwidth}{!}{\input{./figures/07-examples/steane/state-like-derivation.tikz}}\]
\end{proof}

Using this efficient implementation of the $\ket{0}$-like linear map, we can complete the optimised version of Steane-style syndrome extraction:
\begin{proposition}
  \label{prop:floq-steane}
 We can optimise Steane-style syndrome extraction for the Steane code as follows:
  \[\input{./figures/07-examples/steane/steane-code-optimised-statement.tikz}\]
\end{proposition}
\begin{proof}
  \[\input{./figures/07-examples/steane/steane-code-optimised.tikz}\]
\end{proof}

Thus, we have a method to perform Steane-like syndrome extraction for the Steane code using five qubits instead of the previous optimum of eight (seven auxiliary qubits + one flag qubit).
We have also reduced the number of CNOTs required from 18 to 15 and the number of measurements from eight to five.

\paragraph{Simulations}
To validate the performance of our optimised protocol, we conducted numerical simulations using Stim~\cite{gidneyStimFast2021}.
We compared the logical error rate of our optimised Steane method, which uses five ancilla qubits, against the standard Steane-style extraction method requiring eight ancillae.
As these two methods are fault-equivalent, we expect them to perform similarly. 
However, as our optimised method has fewer spacetime locations, we anticipate a slightly lower error rate. 
For details on decoding the optimised Steane-like protocol, see \autoref{appendix:stean-decoding}.

In a simulation run for a single cycle while varying the physical error rate ($\pPhis$), our method shows a consistently lower logical error probability; see \autoref{fig:le-per-pe}.
For this simulation, the noise model assumed $\pTwo = \pSPAM = \pPhis$ and a projected memory error of $\pMem = 0.1 \cdot \pPhis$.
To ensure a consistent number of logical errors and a stable confidence interval, we take $15 / \pPhis^2$ samples for each point.

\begin{figure}[h!]
  \centering
  \includegraphics[width=\columnwidth]{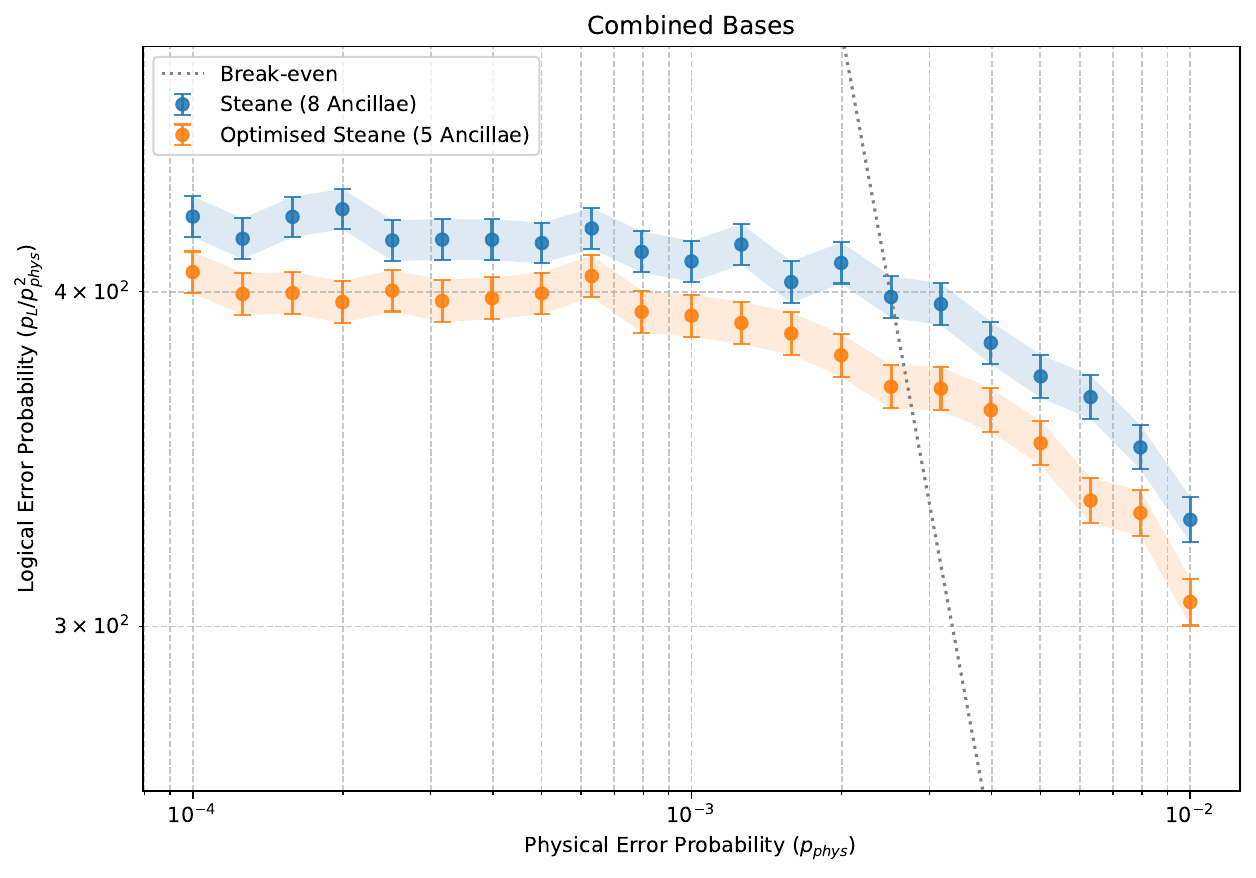}
  \caption{Comparison of the logical error probability ($p_L$) against the physical error probability ($\pPhis$) for a single syndrome extraction cycle.}
  \label{fig:le-per-pe}
\end{figure}

We ran a second simulation with a realistic error model for near-term hardware ($\pTwo = \pSPAM = 10^{-3}$ and $\pMem = 10^{-4}$) while repeating the syndrome measurement and correction cycle ($N$).
The results, shown in \autoref{fig:le-per-cycles}, demonstrate that our optimised method maintains its superior performance while increasing the number of cycles.
Each data point is averaged over $\frac{\pPhis^2}{15N}$ samples to obtain a consistent number of logical errors and a stable confidence interval for each data point.

\begin{figure}[h!]
  \centering
  \includegraphics[width=\columnwidth]{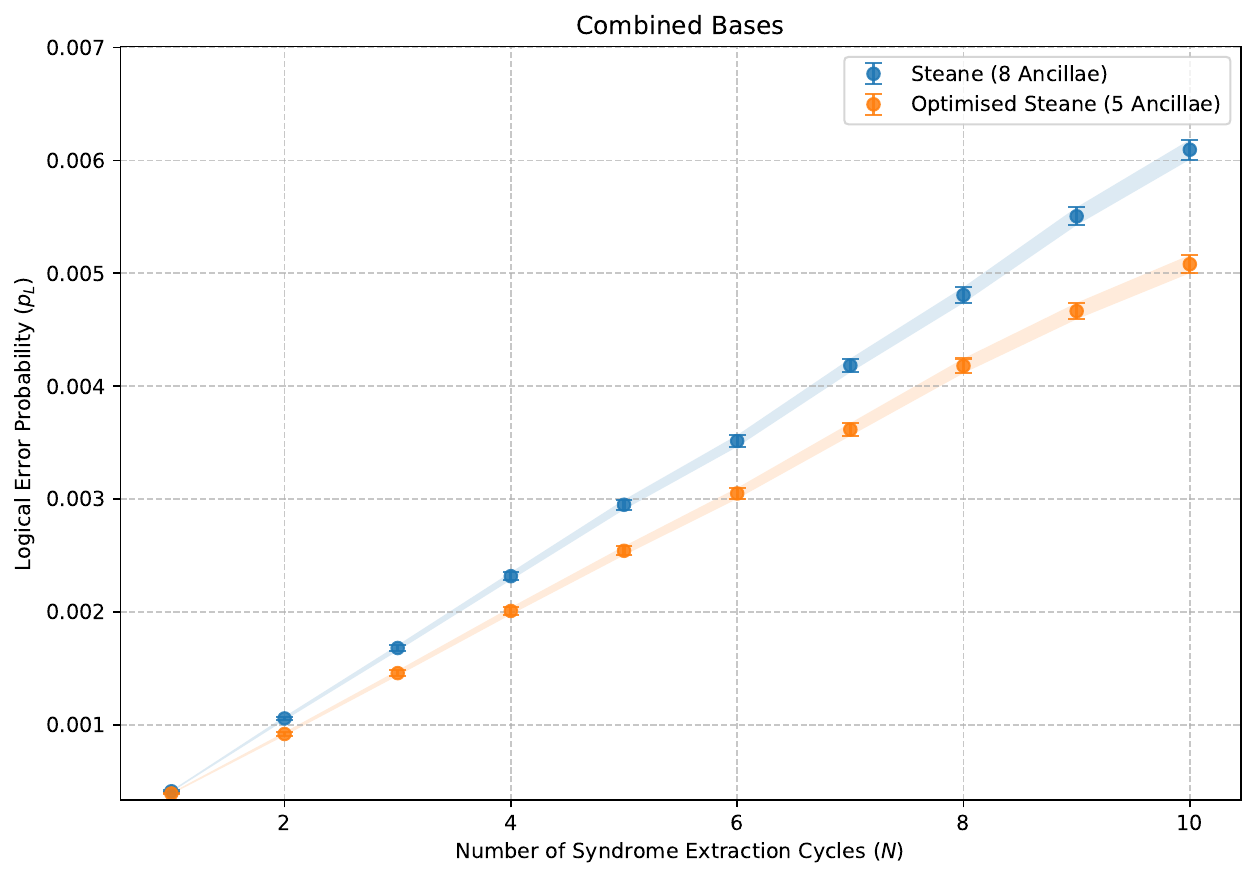}
  \caption{Comparison of the logical error probability against the number of syndrome extraction cycles.}
  \label{fig:le-per-cycles}
\end{figure}

These two simulations exactly match our prediction; asymptotically, both methods scale the same, but our optimised method performs better by a constant offset due to having fewer spacetime locations.
As such, this is a first indication of how fault equivalent rewrites can be used to optimise circuits.

\section{Conclusion}

In this work, we introduced a comprehensive framework for the specification, synthesis, and optimization of fault-tolerant quantum circuits, providing a method to design circuits that are fault-tolerant by construction.
Our approach is centred around \emph{fault equivalence}, a formal concept we introduced to define when two circuits exhibit the same behaviour under the influence of noise.
We demonstrated that this is a foundational and unifying idea, capable of recovering and formalizing many implicit design goals found throughout the literature.

To make this concept practical and computationally tractable, we adapted the ZX calculus.
We developed a methodology for constructing \emph{fault-equivalent representations} that faithfully map noisy quantum circuits into the graphical language.
We identified a set of \emph{fault-equivalent rewrite rules} that, unlike standard ZX rules, provably preserve a circuit's behaviour under noise.
This gives rise to a new paradigm of fault-tolerant circuit design where correctness is guaranteed by construction, avoiding the need for complex, resource-intensive verification after the synthesis is complete.

We demonstrated the practical applicability of this framework through several concrete examples.
We presented simple, intuitive graphical proofs for the correctness of both Shor- and Steane-style syndrome extraction.
Furthermore, we used our fault-equivalent rewrites to synthesize a novel cat state preparation protocol applicable to codes of any distance and to optimize existing syndrome extraction circuits.
Notably, our method yielded a Shor-style implementation that halves the number of auxiliary qubits required and an optimized circuit for the Steane code that reduces the ancilla qubit count from eight to five during syndrome extraction.
In STIM simulation~\parencite{gidneyStimFast2021}, we verified the improved performance of the optimised circuits. 
These novel circuits, achieved through the compositional application of our rewrite rules, are fault-tolerant despite the lack of immediate structural indicators, overcoming the need for transversal constructions.

Looking forward, we anticipate that this perspective can capture and unify many different approaches in fault-tolerant quantum computing.
The language may help identify new structural indicators of fault tolerance beyond established ideas like transversality.
Moreover, it enables the exploration of entirely new circuit implementations with practical hardware benefits;
for instance, our framework was used to derive a teleportation-like syndrome measurement circuit that could aid with routing on chips with limited connectivity.
Ultimately, this work paves the way toward an end-to-end circuit compilation framework, enabling the systematic translation of idealized algorithms into practical, resource-efficient, and verifiably fault-tolerant quantum circuits.

This requires a variety of future work. 
Firstly, all the proofs presented in this work were manually derived.
Automating the rewrite process would be necessary to scale this approach to broader families of circuits.
Secondly, this work does worst-case analyses of circuits, considering the worst-case behaviour of circuits if at most $w$ faults occur.
This is reflected in the fact that we use adversarial noise models, which do not attribute probabilities to faults. 
However, in practice, we are interested in the average-case behaviour of circuits under stochastic noise.
Therefore, in future work, we want to generalise this approach to account for probability distributions over faults. 
Additionally, fault-equivalent rewrites can introduce or remove detecting regions, changing the decoding problem.
An interesting research direction would be to explore the preservation of efficient decodability under (some) fault-equivalent rewrites.
Finally, this framework, like most prior work on spacetime codes~\parencite{baconSparseQuantum2017, gottesmanOpportunitiesChallenges2022,delfosseSpacetimeCodes2023}, only considers Clifford circuits.
While this is an important circuit fragment for fault-tolerant quantum computing and is sufficiently expressive for many important problems, in the long term, we would like to go beyond Clifford circuits.

\section*{Acknowledgements}
We would like to thank Alex Townsend-Teague, Șerban Cercelescu, Razin Shaikh, Daniel Miller, Eric Kuehnke, Simon Burton and Maximilian Rüsch for the useful discussions and their generous feedback on earlier drafts of this paper.
We also thank the Leibniz-Zentrum f\"ur Informatik GmbH for hosting the seminar ``Quantum Error Correction Meets ZX-Calculus''\footnote{\url{https://www.dagstuhl.de/en/seminars/seminar-calendar/seminar-details/25382}} at Schloss Dagstuhl, where we had many useful interactions that helped develop many of the ideas that appear in this paper.
BR thanks Simon Harrison for his generous support for the Wolfson Harrison UK Research Council Quantum Foundation Scholarship.
BP and AK are supported by the Engineering and Physical Sciences Research Council grant number EP/Z002230/1, \enquote{(De)constructing quantum software (DeQS)}.
BR and BP are employed part-time by Quantinuum.
The wording in some sections of this paper has been refined using LLMs.

\printbibliography

\newpage
\appendix

\section{Measurement Outcomes and Fault-equivalence}
\label{appendix:measurement-outcomes}

The idea of equivalence or fault-equivalence is that we can replace one diagram/circuit with an equivalent one and still get the same behaviour. For unitary circuits, the ``same behaviour'' means simply that the LHS and RHS implement the same unitary. However, circuits that include measurements represent non-deterministic processes, where multiple outcomes each correspond to a distinct linear map, or \textit{branch} in which each measurement has resulted in a projection onto the $+1$ or $-1$ eigenspace of the associated Pauli. In this case, we need a way to relate measurement outcomes on the LHS to those on the RHS and establish (fault-)equivalence for all branches of the computation.

In this Appendix, we will show that this is always possible for Clifford circuits and ZX-diagrams without ever making explicit reference to measurement outcomes. That is, (fault-)equivalence on a single branch of a non-deterministic computation implies (fault-)equivalence on all branches. For the sake of brevity, we will demonstrate how this works just using ZX-diagrams, but the principle for circuits is identical.

We noted in Section~\ref{sec:spacetime-locs} that we can treat measurements in circuits as projectors onto the $+1$ eigenspace of a Pauli ``by default'', and if we want to obtain the other projector, we can introduce an anti-commuting Pauli around the measurement. We can treat the presence or absence of such a Pauli as the measurement outcome itself. For example, a single-qubit $Z$-measurement can be written simply as this diagram:
\ctikzfig{appendix/zx-z-effect}
Whereas normally we would interpret this ZX-diagram as post-selection on to $\<0|$, we can treat this as a proper measurement with outcome $k \in \{ 0, 1 \}$ by introducing an anti-commuting Pauli $X$ onto the wire directly before the measurement iff $k = 1$:
\ctikzfig{appendix/zx-z-measure}

For a multi-qubit Pauli measurement, we can flip the measurement outcome by applying an anti-commuting Pauli before and after the projector. In the case of a Pauli $ZZZZ$ measurement, we can introduce a Pauli $X$ on any of the 4 qubits, as we did in the circuit example in Section~\ref{sec:spacetime-locs}:
\ctikzfig{appendix/zx-zzzz-measure}
Equivalently, we can place the Pauli $X$ on one of the internal wires involved in the measurement projector:
\begin{equation}\label{eq:zzzz-meas}
\input{./figures/appendix/zx-zzzz-measure-edge.tikz}
\end{equation}

In Section~\ref{sec:shor-example}, the main part of showing the correctness of Shor-style syndrome extraction we establishing the following fault-equivalence:
\begin{equation}\label{eq:shor-fault-equiv}
\input{./figures/appendix/shor-fault-equiv.tikz}
\end{equation}
The ZX-diagram on the LHS is interpreted as a circuit with one measurement, whereas the diagram on the RHS has four measurements. So, in order to interpret these two diagrams as equivalent computations, we should have a recipe for interpreting measurement outcomes on the LHS as measurement outcomes on the RHS and vice-versa.

Let's first see how to do this for the example above, then give the general recipe. First, we can decorate the diagrams on the LHS and RHS of \eqref{eq:shor-fault-equiv} with Paulis that represent every possible measurement outcome:
\ctikzfig{appendix/shor-equiv-dec}
We then push those Paulis from the interior wires on to the boundary. We already did this for the LHS in equation~\eqref{eq:zzzz-meas}, so we just need to do it for the RHS:
\ctikzfig{appendix/shor-fault-equiv-rhs}
Hence, we can see that on the RHS, even though there are four measurements yielding 16 possible outcomes, there are only two genuinely different branches of the computation, corresponding to the parity $k = k_1 \oplus k_2 \oplus k_3 \oplus k_4$. These are exactly the two possibilities on the LHS, so we can treat these two computations as equivalent.

Also, note that we only proved fault-equivalence for a single measurement outcome, notably the all-$0$ outcome. However, once we push outcomes to the boundary, it is clear that this suffices to provide fault-equivalence for all outcomes:
\ctikzfig{appendix/shor-fault-equiv-dec}

More formally, for any Clifford linear map $L$, we can define a \textit{branch} of $L$ as a pair of Pauli operators $(\vec P, \vec Q)$ up to the equivalence: $(\vec P, \vec Q) \sim (\vec P', \vec Q')$ iff there exists a Pauli correction $\vec C$ such that $\vec Q L \vec P = \vec C \vec Q' L \vec P'$.
Since the set of branches only depends on the linear map of a given diagram and not a particular representation, any fault-equivalence gives a correspondence between the measurement outcomes on the LHS and the RHS, up to a possible Pauli correction on the outputs.

For any Clifford ZX-diagram with decorated with measurement outcomes, we can efficiently compute the branch corresponding to a given outcome by pushing the associated Pauli operator to the boundary. For any set of measurement outcomes that happen with non-zero probability, we can always do this thanks to the following proposition.

\begin{proposition}
For a diagram $D$ and Pauli $P$, then if $D^P \neq 0$ then there exist Paulis $P_1, P_2$ such that $D^P = P_2 D^I P_1$.
\end{proposition}

\begin{proof}
We can find $P_1$ and $P_2$ constructively as follows. Fix a generating set of Pauli webs on the diagram $D$ such that the restriction of each web with with support on the boundary forms an independent set of Paulis operators. This restriction to the boundary of these Pauli webs (paying attention to signs) forms a complete set of spacetime stabilisers $\mathcal S$ for $D^I$. That is, they form a maximal set of Pauli operators $Q_1 \otimes Q_2$ such that $Q_2 D^I Q_1 = D^I$. Now, the Pauli $P$ defined over the wires of $D$ either commutes or anti-commutes with each element in the generating set of Pauli webs, which corresponds to flipping the sign its associated spacetime stabiliser in $\mathcal S$, obtaining a new set of spacetime stabilisers $\mathcal S'$ for $D^P$. Since $D^P$ is non-zero, it is uniquely fixed by $\mathcal S'$, and since we chose the elements in $\mathcal S$ to be independent, there must exist a Pauli $P_1 \otimes P_2$ supported on the boundary of $D$ that sends $\mathcal S$ to $\mathcal S'$ under conjugation. This implies that $D^P$ and $P_2 D^I P_1$ both have $\mathcal S'$ as a maximal set of spacetime stabilisers, hence $D^P = P_2 D^I P_1$.
\end{proof}

In the Shor-style syndrome extraction example, we could take the final Pauli correction to be trivial. However, in general we could obtain non-trivial Pauli corrections. For example, consider this implementation of a $ZZZZ$ measurement with a pair of $XX$ measurments on ancilla qubits, as we did in the proof of Proposition~\ref{prop:shor-optimised}:
\begin{equation}
\input{./figures/appendix/shor-optimised-app.tikz}
\end{equation}
We already identified branches on the LHS. For the RHS, we can decorate each of the four measurements with outcomes $k_1, k_2, k_3, k_4$:
\ctikzfig{appendix/shor-optimised-app-outcomes1}
If $k_3 \oplus k_4 \neq 0$, this diagram goes to $0$, due to the fact that $k_3$ and $k_4$ lie in the same detecting region. If $k_3 \oplus k_4 = 0$, then we can push the outcomes out to the boundary as follows:
\ctikzfig{appendix/shor-optimised-app-outcomes}
As we can see, the branch is fixed by the parity $k_1 \oplus k_2$ of the two $Z$ measurements and there is a Pauli correction on the outputs, which depends on the outcomes of the two $XX$ measurements.

This ``pushing out'' of outcomes, as well as computing equivalence on the boundary, up to Pauli corrections, amounts to solving systems of $\mathbb F_2$-linear equations, so it can be done efficiently for arbitrary Clifford circuits or ZX-diagrams.

\section{Proofs of ZX rewrites}

\subsection{General propositions}
In this section, we will prove some helper propositions for ZX rewrites. 
Note that these rewrites only preserve the underlying linear map, not fault equivalence.
\begin{proposition}
  \label{lem:pi-state}
  \[
 \input{./figures/appendix/zx-rewrites/pi-stabilized-state.tikz}
  \]
\end{proposition}
\begin{proof}
  \[
 \input{./figures/appendix/zx-rewrites/pi-stabilized-state-proof.tikz}
  \]
\end{proof}

\begin{proposition}
  \label{lem:fire-spider}
  \[
 \input{./figures/appendix/zx-rewrites/fire-spider.tikz}
  \]
\end{proposition}
\begin{proof}
  \[
 \input{./figures/appendix/zx-rewrites/fire-spider-proof.tikz}
  \]
\end{proof}

\begin{proposition}
  \label{lem:hh}
  \[
 \input{./figures/appendix/zx-rewrites/hadamard-hadamard.tikz}
  \]
\end{proposition}
\begin{proof}
  \[
 \input{./figures/appendix/zx-rewrites/hadamard-hadamard-proof.tikz}
  \]
\end{proof}

\subsection{Fault-Equivalent Rewrites}\label{subsec:fault-equivalent-rewrite-proofs}

In this section, we will prove some additional fault-equivalent rewrites.
For this, we first prove the following:

\begin{proposition}[Fault Equivalence via Pushing to Boundaries]
  \label{lem:fault-equiv-boundaries}
 Let $r \colon D_1 \to D_2$ be a rewrite rule.
 Let the respective edge sets of the involved diagrams be $E(D_1) = E_{I,1} \sqcup E_{B,1}$ and $E(D_2) = E_{I,2} \sqcup E_{B,2}$, that is, the disjoint set of internal and boundary edges for both.
 Since both diagrams must have the same number of inputs and outputs, there exists a unique bijection $\phi \colon E_{B,1} \to E_{B,2}$ between their boundaries.

 We say a diagram has the \emph{$w$-fault push-out property} if any undetectable fault of weight less than $w$ on its internal edges can be replaced by a logically equivalent fault on its boundaries with non-increasing weight.

 The rewrite $r$ is fault-equivalent if both $D_1$ and $D_2$ have the boundary push-out property.
\end{proposition}
\begin{proof}
 To prove $w$-fault equivalence, we must show for both directions that an undetectable fault of weight less than $w$ on one diagram corresponds to an equivalent fault on the other with non-increasing weight.
 The argument is symmetric; we show it explicitly for the direction $D_1 \to D_2$.

 Let $F_1$ be an arbitrary undetectable fault on $D_1$ with $wt(F_1) < w$.
 Our goal is to construct a corresponding fault $F_2$ on $D_2$ such that $wt(F_2) \leq wt(F_1)$.
 This is achieved in two stages: first, pushing the fault to the boundary of $D_1$, and then mapping it to $D_2$.

  \paragraph{Pushing the fault to the boundary of $D_1$.}
 We decompose the initial fault into its internal and boundary parts, $F_1 = F_{I,1} \cdot F_{B,1}$.
 The proposition's premise states that $D_1$ has the boundary push-out property.
 Applying this to the internal part $F_{I,1}$, there must exist an equivalent fault $F'_{B,1}$ with support only on the boundary $E_{B,1}$ and with $wt(F'_{B,1}) \le wt(F_{I,1})$.

 Therefore, the original fault $F_1$ is logically equivalent to a new fault on $D_1$ that is purely boundary-based:
  \[
 {D_1^{F_1}} = {D_1^{F_{I,1} \cdot F_{B,1}}} = {D_1^{F'_{B,1} \cdot F_{B,1}}}.
  \]
 Let this total boundary fault on $D_1$ be $F'_{1} = F'_{B,1} \cdot F_{B,1}$.
 Its weight is bounded by $wt(F'_{1}) \le wt(F_1)$.

  \paragraph{Mapping the boundary fault to $D_2$.}
 We have found a fault $F'_{1}$ that is equivalent to our original fault $F_1$ but is supported only on the boundary of $D_1$.
 Since the boundaries of $D_1$ and $D_2$ are isomorphic via $\phi$, we can map this boundary fault to $D_2$ to get a logically equivalent operation.
 We define our target fault on $D_2$ as the image of $F'_{1}$ under this map, $F_2 \coloneqq \phi(F'_{1})$.
 By this construction, the interpretation of $D_1^{F_1}$ is the same as that of $D_2^{F_2}$.
 The weight of this new fault is $wt(F_2) = wt(\phi(F'_{1})) = wt(F'_{1})$, which we have already shown to be less than or equal to $wt(F_1)$.

 We have thus constructed a fault $F_2$ on $D_2$ that is logically equivalent to $F_1$ with non-increasing weight.
 The argument for the reverse direction ($D_2 \to D_1$) is identical, using the premise that $D_2$ also has the boundary push-out property.
 Therefore, the rewrite is fault-equivalent.
\end{proof}

\begin{proposition}
  \label{lem:fe-pi-copy}
 The \TextPiCommute rule is a fault-equivalent rewrite:
  \begin{gather}
 \tag{$\textsc{Pi-Copy}_{\text{fe}}$}\label{fe-pi-copy}\refstepcounter{equation}
 \input{./figures/appendix/zx-rewrites/FE-pi-copy.tikz}
  \end{gather}
\end{proposition}
\begin{proof}
 By \autoref{lem:fault-equiv-boundaries}, we only have to show that any fault on internal edges can be pushed out.

 The diagram on the left has a single internal edge, while the one on the right has two.
 Each of these internal edges is connected to a boundary edge by a 2-legged $\pi$ spider.
 Consequently, to prove the fault equivalence, we only need to show that an arbitrary edge flip can be propagated through such a spider, as shown below:
  \[
 \input{./figures/appendix/zx-rewrites/FE-pi-copy-proof.tikz}
  \]
\end{proof}

\begin{proposition}
  \label{lem:fe-phase-unfuse}
  \begin{gather}
 \tag{$\textsc{Unfuse}_{\text{fe}}$}\label{phase-unfuse-fe}\refstepcounter{equation}
 \input{./figures/appendix/zx-rewrites/FE-phase-unfuse.tikz}
  \end{gather}
\end{proposition}
\begin{proof}
 By \autoref{lem:fault-equiv-boundaries}, it suffices to consider the case where the fault is on an internal edge.
 The diagram for the proposition contains a single internal edge connected to a boundary edge via a 2-legged $k \frac{\pi}{2}$-phased spider.
 Fault-equivalence is proven by showing that an arbitrary edge flip can be propagated through such a spider, as shown in the following rewrite:
  \[
 \input{./figures/appendix/zx-rewrites/FE-phase-unfuse-proof.tikz}
  \]
\end{proof}

\begin{proposition}
  \label{lem:fe-pi-pi-id}
  \[
 \input{./figures/appendix/zx-rewrites/FE-pi-pi-id.tikz}
  \]
\end{proposition}
\begin{proof}
 The proof follows by the same reasoning as in \autoref{lem:fe-pi-copy}, simply substituting the red spider with a green $\ pi$-phased spider.
 The necessary rewrite is given as follows:
  \[
 \input{./figures/appendix/zx-rewrites/FE-pi-pi-id-proof.tikz}
  \]
\end{proof}

\begin{proposition}
  \label{lem:fe-perfect-fuse}
  \[\input{./figures/appendix/zx-rewrites/FE-idealised-fuse.tikz}\]
\end{proposition}
\begin{proof}
 As neither diagram allows faults on internal edges, by \autoref{lem:fault-equiv-boundaries}, the two diagrams are fault-equivalent.
\end{proof}

\begin{proposition}
  \label{lem:fe-copy}
 The following rule is a fault-equivalent rewrite:
  \[
 \input{./figures/appendix/zx-rewrites/FE-copy.tikz}
  \]
\end{proposition}
\begin{proof}
 As neither diagram allows faults on internal edges, by \autoref{lem:fault-equiv-boundaries}, the two diagrams are fault-equivalent.
\end{proof}

\begin{proposition}
  \label{lem:fe-cat-xs}
  \[
 \input{./figures/appendix/zx-rewrites/FE-cat-Xs.tikz}
  \]
\end{proposition}
\begin{proof}
 For an $n$-legged cat state, the diagram features $n$ internal edges.
 To prove the identity, we must show that a fault on any of these internal edges can be propagated past the $2$-legged $\pi$-spider that connects it to the boundary.
 This process can be applied to each leg independently.
 Therefore, the required rewrite for each instance is precisely the one provided in the proof of \autoref{lem:fe-pi-copy}.
\end{proof}

\inductiveDistRewrite*
\begin{proof}
 By \autoref{lem:fault-equiv-boundaries}, to prove fault equivalence, we only have to show that any undetectable fault on the internal edges of the diagram on the right-hand side can be pushed to the boundary.

 The proof strategy is to analyze the set of all undetectable internal faults.
 Such faults are identified as the null space of a parity-check matrix $P$ derived from the diagram's detecting regions.
 We first construct an explicit basis for this null space, representing all fundamental undetectable faults.
 Then, we show that each basis vector can be pushed to the boundary without increasing its weight.
 Finally, we argue that an arbitrary undetectable fault, which is just a linear combination of these basis vectors, can also be pushed to the boundary without an increase in weight.
 This directly satisfies the conditions for fault equivalence, ensuring that no undetectable logical error is introduced or removed by the rewrite.

  \paragraph{Defining the parity-check matrix}
 In total, there are $3n$ internal edges, and thus, potential independent edge flip locations to consider.
 As $Z$ edge flips commute with all the spiders in the diagram and can therefore be freely pushed out, and we decompose a $Y$ edge flip as the product of an $X$ and a $Z$ edge flip, we only have to consider $X$ edge flips.

 A green detecting region has to include an even number of the $3$-legged spider pairs at the boundary.
 We can construct a basis of size $n - 1$ for all green detecting regions by considering the regions that include the first outer spider and the $i$-th outer spider for $i \in [2, n]$:
  \[
 \input{./figures/appendix/zx-rewrites/detecting-regions.tikz}
  \]

 Let us construct an $X$ parity-check matrix $P_X \in F_2^{(n-1) \times 3n}$, where each of the $n-1$ rows corresponds to a detecting region and each of the $3n$ columns corresponds to an internal edge.
 Then, for a fault vector $\vec v \in F_2^{3n}$, the syndrome is calculated as $P_X \vec v$.

 As there are $n - 1$ detecting regions forming the parity-check matrix $P_X$, it has a null space with dimension $3n - \text{rank}(P_X) = 3n - \left( n - 1 \right) = 2n + 1$.
 We now give a basis for this null space and prove that any basis vector can be pushed out without increasing its weight.

  \paragraph{Giving a basis for the null space of $P_X$}
 First, we consider the fault $\vec{r_0}$:
  \[
 \input{./figures/appendix/zx-rewrites/error-2.tikz}
  \]
 Since each detecting region overlaps with this fault on an even number of legs, it is not detectable.
 However, this fault is trivial, as firing $s$ removes all edge flips.

 The next set of faults $\vec{r_1}, \dots, \vec{r_n}$ consists of faults between the internal spiders and the 3-legged spider pair:
  \begin{equation}
    \label{eq:rec-cat-push-2}
 \input{./figures/appendix/zx-rewrites/error-1.tikz}
  \end{equation}
 Each detecting region either highlights both internal edges of an outer spider or neither, therefore, this error is not detectable.
 However, by firing the pair of $3$-legged spiders, we can push the error to the boundary edges.

 Lastly, we consider faults $\vec{r_{n + 1}}, \dots, \vec{r_{2n}}$ incorporating an edge flip between the $3$-legged spiders and another on the internal edge of one of the $3$-legged spiders:
  \begin{equation}
    \label{eq:rec-cat-push-3}
 \input{./figures/appendix/zx-rewrites/error-3.tikz}
  \end{equation}
 Again, each detecting region either includes $2$ or $0$ of these faults, making them undetectable.
 Nevertheless, firing corresponding spider $s$ pushes the fault to the boundary.

 We now have $1 + n + n = 2n + 1$ independent faults, thereby spanning the null space of $P_X$.
 Therefore, any undetectable fault $\vec{v}$ (i.e.\@ any vector in the null space of $P_X$) can be expressed as a linear combination of these basis vectors spanning $P_X$, i.e.\@ $\vec{v} = \sum_{i = 0}^{2n} \lambda_{r_i} \vec{r_i}$, where $\lambda_{r_i} \in F_2$.

  \paragraph{Pushing faults out}
 Let us consider an arbitrary fault of the form $\vec{v} = \lambda_{r_1} \vec{r_1} + \dots +  \lambda_{r_{2n}} \vec{r_{2n}}$.
 Let us define the following values:
  \[
 A \coloneqq \sum_{i = 1}^{n} \lambda_i
 \qquad\qquad\qquad
 B \coloneqq \sum_{i = n + 1}^{2n} \lambda_i
 \qquad\qquad\qquad
 C \coloneqq \sum_{i = 1}^{n} \lambda_i \lambda_{n + i}
  \]
 That is, in $\vec{v}$, $A$ is the number of basis elements of the form \autoref{eq:rec-cat-push-2}, $B$ is the number of basis elements of the form \autoref{eq:rec-cat-push-3}, and $C$ quantifies when both of these errors act on the same 3-legged spider pair.

 With this, we can quantify the weight of $\vec v$ as $wt(\vec v) = 2A + 2B - 2C$.
 Then, pushing out these faults according to \autoref{eq:rec-cat-push-2} and \autoref{eq:rec-cat-push-3} results in a fault $\vec{v'}$ with $wt(\vec{v'}) = 2A + B - 2C$.
 Consequently, the condition of pushing out to create a fault with a smaller weight is satisfied.

 Now, let us consider a fault of the form $\vec{w} = \vec{r_0} + \lambda_{r_1} \vec{r_1} + \dots +  \lambda_{r_{2n}} \vec{r_{2n}}$, and define $A$, $B$, and $C$ similarly to the above.
 Then, the weight of the error is $wt(\vec{w}) = n + 2C$.
 This is because the weight of $\vec{r_0}$ is $n$, and then if a pair of 3-legged spiders at the boundary has an error of the form of
  \begin{itemize}
    \item only \autoref{eq:rec-cat-push-2} or only \autoref{eq:rec-cat-push-3}, an edge flip is added, and one is removed, leaving the overall weight unchanged;
    \item both \autoref{eq:rec-cat-push-2} and \autoref{eq:rec-cat-push-3} at the same time, two edge flips are added to the weight.
  \end{itemize}

 Pushing each basis fault out individually still results in a fault with weight $wt(\vec{w'}) = 2A + B - 2C$.
 Now, if $wt(\vec{w'}) \leq n$, then this pushed-out fault is adequate.
 Otherwise, $wt(\vec{w'}) > n$, in which case, we fire the spiders according to the Pauli-web that witnesses the all $X$ stabiliser.
 Effectively, this flips the bits on each of the outputs.
 This results in a fault $\vec{w''}$, such that $wt(\vec{w''}) = 2n - wt(\vec{w'}) < n$, finalizing the proof.

 Since in both cases, any undetectable fault can be pushed to the boundary with a weight less than or equal to its original weight, the diagram has the boundary push-out property.
 By \autoref{lem:fault-equiv-boundaries}, this completes the proof.
\end{proof}

\inductiveDistRewriteW*
\begin{proof}
 While this statement looks similar to the previous one, our proof strategy will be vastly different.
 As we only have to prove $w$-fault equivalence, the only faults we have to care about are undetectable faults of weight less than $w$.
 This means we can say a lot more about the structure of the undetectable faults.

 In particular, naming the spiders, we have:
  \[\input{./figures/04-zx/fault-equiv-rewrites/recursive-spider-w-proof-1.tikz}\]
 Detecting regions on this diagram consists of even many subregions that highlight edges $(s_0, s_i), (s_i, s'_i), (s'_i, s'_0)$, e.g.\@ even many of the following highlightings:
  \[\input{./figures/04-zx/fault-equiv-rewrites/recursive-spider-w-proof-2.tikz}\]

 Let us consider a fault $F$ of weight less than $w$.
 As there are $w$ subregions of the form above, we know that there must be at least one subregion $(s_0, s_i), (s_i, s'_i), (s'_i, s'_0)$ that does not have an edge flip in $F$.
 But then, since $F$ is assumed to be undetectable, all other subregions $(s_0, s_j), (s_j, s'_j), (s'_j, s'_0)$ must have even overlap with $F$.
 Otherwise, the detecting region $(s_0, s_i), (s_i, s'_i), (s'_i, s'_0), (s_0, s_j), (s_j, s'_j), (s'_j, s'_0)$ would have odd overlap with $F$, violating our assumption that $F$ is undetectable.

 But then, either $F$ does not overlap with the $j$-th subregion, or it overlaps as follows:
  \[\input{./figures/04-zx/fault-equiv-rewrites/recursive-spider-w-proof-3.tikz}\]
 As for each subregion, we can always push the edge flips to the boundary without increasing the weight of the fault, this means that $F$ must obey the push-out property, completing the proof.
\end{proof}

\section{Steane-style syndrome extraction --- Additional details}
\subsection{Steane-style syndrome extraction --- Full proof}
\label{appendix:steane-proof}
In the proof of correctness for Steane-style syndrome extraction, we glossed over a subtle but important step.
In particular, the following two diagrams are \emph{not} fault-equivalent in general:
\[
\begin{tikzpicture}
	\begin{pgfonlayer}{nodelayer}
		\node [style=none] (0) at (0, 0) {$\NotFaultEq$};
		\node [style=X dot] (2) at (3.5, -2) {};
		\node [style=none] (3) at (1.5, 1.5) {};
		\node [style=none] (4) at (1.5, 0.5) {};
		\node [style=none] (5) at (1.5, -1.5) {};
		\node [style=Z dot] (6) at (2.5, 1) {};
		\node [style=Z dot] (7) at (2.5, 0) {};
		\node [style=Z dot] (8) at (2.5, -1) {};
		\node [style=X dot] (9) at (3.5, 0) {};
		\node [style=X dot] (10) at (3.5, 1) {};
		\node [style=none, rotate=90] (11) at (4, -0.5) {...};
		\node [style=none, rotate=90] (12) at (1.75, -0.5) {...};
		\node [style=Z dot] (13) at (6.75, -2) {};
		\node [style=X dot] (14) at (5.5, 1) {};
		\node [style=X dot] (15) at (5.5, 0) {};
		\node [style=X dot] (16) at (5.5, -1) {};
		\node [style=none] (17) at (8.25, 1.5) {};
		\node [style=none] (18) at (8.25, 0.5) {};
		\node [style=none] (19) at (8.25, -1.5) {};
		\node [style=Z dot] (20) at (6.75, 0) {};
		\node [style=Z dot] (21) at (6.75, 1) {};
		\node [style=none, rotate=90] (22) at (7.25, -0.5) {...};
		\node [style=X dot] (23) at (4, 1.5) {};
		\node [style=X dot] (24) at (4, 0.5) {};
		\node [style=X dot] (25) at (4, -1.5) {};
		\node [style=Z dot] (26) at (7.25, 0.5) {};
		\node [style=Z dot] (27) at (7.25, 1.5) {};
		\node [style=Z dot] (28) at (7.25, -1.5) {};
		\node [style=X dot] (29) at (-6.25, -2) {};
		\node [style=none] (30) at (-8.25, 1.5) {};
		\node [style=none] (31) at (-8.25, 0.5) {};
		\node [style=none] (32) at (-8.25, -1.5) {};
		\node [style=Z dot] (33) at (-7.25, 1) {};
		\node [style=Z dot] (34) at (-7.25, 0) {};
		\node [style=Z dot] (35) at (-7.25, -1) {};
		\node [style=X dot] (36) at (-6.25, 0) {};
		\node [style=X dot] (37) at (-6.25, 1) {};
		\node [style=none, rotate=90] (38) at (-5.75, -0.5) {...};
		\node [style=none, rotate=90] (39) at (-8, -0.5) {...};
		\node [style=Z dot] (40) at (-3, -2) {};
		\node [style=X dot] (41) at (-4.25, 1) {};
		\node [style=X dot] (42) at (-4.25, 0) {};
		\node [style=X dot] (43) at (-4.25, -1) {};
		\node [style=none] (44) at (-1.5, 1.5) {};
		\node [style=none] (45) at (-1.5, 0.5) {};
		\node [style=none] (46) at (-1.5, -1.5) {};
		\node [style=Z dot] (47) at (-3, 0) {};
		\node [style=Z dot] (48) at (-3, 1) {};
		\node [style=none, rotate=90] (49) at (-2.5, -0.5) {...};
		\node [style=X dot] (50) at (-5.75, 1.5) {};
		\node [style=X dot] (51) at (-5.75, 0.5) {};
		\node [style=X dot] (52) at (-5.75, -1.5) {};
		\node [style=Z dot] (53) at (-2.5, 0.5) {};
		\node [style=Z dot] (54) at (-2.5, 1.5) {};
		\node [style=Z dot] (55) at (-2.5, -1.5) {};
	\end{pgfonlayer}
	\begin{pgfonlayer}{edgelayer}
		\draw [style=fault-free] (8) to (2);
		\draw [style=fault-free] (8) to (9);
		\draw [style=fault-free] (16) to (13);
		\draw [style=fault-free] (21) to (16);
		\draw [style=fault-free] (10) to (7);
		\draw [style=fault-free] (7) to (2);
		\draw [style=fault-free] (10) to (6);
		\draw [style=fault-free] (6) to (9);
		\draw [style=fault-free] (6) to (2);
		\draw [style=fault-free] (20) to (15);
		\draw [style=fault-free] (21) to (15);
		\draw [style=fault-free] (21) to (14);
		\draw [style=fault-free] (14) to (13);
		\draw (25) to (2);
		\draw (24) to (9);
		\draw (23) to (10);
		\draw (3.center) to (23);
		\draw (4.center) to (24);
		\draw (5.center) to (25);
		\draw (28) to (13);
		\draw (25) to (28);
		\draw (24) to (26);
		\draw (23) to (27);
		\draw (27) to (21);
		\draw (26) to (20);
		\draw (27) to (17.center);
		\draw (26) to (18.center);
		\draw (28) to (19.center);
		\draw [style=fault-free] (35) to (29);
		\draw [style=fault-free] (35) to (36);
		\draw [style=fault-free] (43) to (40);
		\draw [style=fault-free] (48) to (43);
		\draw [style=fault-free] (37) to (34);
		\draw [style=fault-free] (34) to (29);
		\draw [style=fault-free] (37) to (33);
		\draw [style=fault-free] (33) to (36);
		\draw [style=fault-free] (33) to (29);
		\draw [style=fault-free] (47) to (42);
		\draw [style=fault-free] (48) to (42);
		\draw [style=fault-free] (48) to (41);
		\draw [style=fault-free] (41) to (40);
		\draw [style=fault-free] (52) to (29);
		\draw [style=fault-free] (51) to (36);
		\draw [style=fault-free] (50) to (37);
		\draw (30.center) to (50);
		\draw (31.center) to (51);
		\draw (32.center) to (52);
		\draw [style=fault-free] (55) to (40);
		\draw [style=fault-free] (52) to (55);
		\draw [style=fault-free] (51) to (53);
		\draw [style=fault-free] (50) to (54);
		\draw [style=fault-free] (54) to (48);
		\draw [style=fault-free] (53) to (47);
		\draw (54) to (44.center);
		\draw (53) to (45.center);
		\draw (55) to (46.center);
	\end{pgfonlayer}
\end{tikzpicture}

\]
For example, a $Y$ fault on the top wire between the two measurements propagates to an $X$ error on the input and a $Z$ error on the output:
\[
\begin{tikzpicture}
	\begin{pgfonlayer}{nodelayer}
		\node [style=X dot] (1) at (-7, -2) {};
		\node [style=none] (2) at (-9, 1.75) {};
		\node [style=none] (3) at (-9, 0.5) {};
		\node [style=none] (4) at (-9, -1.5) {};
		\node [style=X dot] (8) at (-7, 0) {};
		\node [style=X dot] (9) at (-7, 1) {};
		\node [style=none, rotate=90] (10) at (-6.5, -0.5) {...};
		\node [style=none, rotate=90] (11) at (-8.75, -0.5) {...};
		\node [style=Z dot] (12) at (-3.75, -2) {};
		\node [style=none] (16) at (-2.25, 1.75) {};
		\node [style=none] (17) at (-2.25, 0.5) {};
		\node [style=none] (18) at (-2.25, -1.5) {};
		\node [style=Z dot] (19) at (-3.75, 0) {};
		\node [style=Z dot] (20) at (-3.75, 1) {};
		\node [style=none, rotate=90] (21) at (-3.25, -0.5) {...};
		\node [style=X dot] (22) at (-6.5, 1.75) {};
		\node [style=X dot] (23) at (-6.5, 0.5) {};
		\node [style=X dot] (24) at (-6.5, -1.5) {};
		\node [style=Z dot] (25) at (-3.25, 0.5) {};
		\node [style=Z dot] (26) at (-3.25, 1.75) {};
		\node [style=Z dot] (27) at (-3.25, -1.5) {};
		\node [style=X dot] (28) at (13, -1.75) {};
		\node [style=none] (29) at (11, 2) {};
		\node [style=none] (30) at (11, 0.75) {};
		\node [style=none] (31) at (11, -1.25) {};
		\node [style=X dot] (35) at (13, 0.25) {};
		\node [style=X dot] (36) at (13, 1.25) {};
		\node [style=none, rotate=90] (37) at (13.5, -0.25) {...};
		\node [style=none, rotate=90] (38) at (11.25, -0.25) {...};
		\node [style=Z dot] (39) at (16.25, -1.75) {};
		\node [style=none] (43) at (18.25, 2) {};
		\node [style=none] (44) at (18.25, 0.75) {};
		\node [style=none] (45) at (18.25, -1.25) {};
		\node [style=Z dot] (46) at (16.25, 0.25) {};
		\node [style=Z dot] (47) at (16.25, 1.25) {};
		\node [style=none, rotate=90] (48) at (16.75, -0.25) {...};
		\node [style=X dot] (49) at (13.5, 2) {};
		\node [style=X dot] (50) at (13.5, 0.75) {};
		\node [style=X dot] (51) at (13.5, -1.25) {};
		\node [style=Z dot] (52) at (16.75, 0.75) {};
		\node [style=Z dot] (53) at (16.75, 2) {};
		\node [style=Z dot] (54) at (16.75, -1.25) {};
		\node [style=X dot] (55) at (2.75, -2) {};
		\node [style=none] (56) at (0.75, 1.75) {};
		\node [style=none] (57) at (0.75, 0.5) {};
		\node [style=none] (58) at (0.75, -1.5) {};
		\node [style=X dot] (62) at (2.75, 0) {};
		\node [style=X dot] (63) at (2.75, 1) {};
		\node [style=none, rotate=90] (64) at (3.25, -0.5) {...};
		\node [style=none, rotate=90] (65) at (1, -0.5) {...};
		\node [style=Z dot] (66) at (6, -2) {};
		\node [style=none] (70) at (8, 1.75) {};
		\node [style=none] (71) at (8, 0.5) {};
		\node [style=none] (72) at (8, -1.5) {};
		\node [style=Z dot] (73) at (6, 0) {};
		\node [style=Z dot] (74) at (6, 1) {};
		\node [style=none, rotate=90] (75) at (6.5, -0.5) {...};
		\node [style=X dot] (76) at (3.25, 1.75) {};
		\node [style=X dot] (77) at (3.25, 0.5) {};
		\node [style=X dot] (78) at (3.25, -1.5) {};
		\node [style=Z dot] (79) at (6.5, 0.5) {};
		\node [style=Z dot] (80) at (6.5, 1.75) {};
		\node [style=Z dot] (81) at (6.5, -1.5) {};
		\node [style=none] (82) at (-0.75, 0) {=};
		\node [style=X dot] (83) at (-5.25, 1.75) {\tiny $\pi$};
		\node [style=Z dot] (84) at (-4.5, 1.75) {\tiny $\pi$};
		\node [style=X dot] (85) at (1.75, 1.75) {\tiny $\pi$};
		\node [style=Z dot] (86) at (7.25, 1.75) {\tiny $\pi$};
		\node [style=none] (87) at (9.5, 0) {=};
		\node [style=X dot] (88) at (12, 2) {\tiny $\pi$};
		\node [style=Z dot] (89) at (17.5, 2) {\tiny $\pi$};
		\node [style=Z dot] (90) at (-8, 1) {};
		\node [style=Z dot] (91) at (-8, 0) {};
		\node [style=Z dot] (92) at (-8, -1) {};
		\node [style=Z dot] (93) at (1.75, 1) {};
		\node [style=Z dot] (94) at (1.75, 0) {};
		\node [style=Z dot] (95) at (1.75, -1) {};
		\node [style=Z dot] (96) at (12, 1.25) {};
		\node [style=Z dot] (97) at (12, 0.25) {};
		\node [style=Z dot] (98) at (12, -0.75) {};
		\node [style=X dot] (99) at (15, 1.25) {};
		\node [style=X dot] (100) at (15, 0.25) {};
		\node [style=X dot] (101) at (15, -0.75) {};
		\node [style=X dot] (102) at (4.75, 1) {};
		\node [style=X dot] (103) at (4.75, 0) {};
		\node [style=X dot] (104) at (4.75, -1) {};
		\node [style=X dot] (105) at (-5, 1) {};
		\node [style=X dot] (106) at (-5, 0) {};
		\node [style=X dot] (107) at (-5, -1) {};
	\end{pgfonlayer}
	\begin{pgfonlayer}{edgelayer}
		\draw (24) to (1);
		\draw (23) to (8);
		\draw (22) to (9);
		\draw (2.center) to (22);
		\draw (3.center) to (23);
		\draw (4.center) to (24);
		\draw (27) to (12);
		\draw (24) to (27);
		\draw (23) to (25);
		\draw [in=180, out=0] (22) to (26);
		\draw (26) to (20);
		\draw (25) to (19);
		\draw (26) to (16.center);
		\draw (25) to (17.center);
		\draw (27) to (18.center);
		\draw [style=fault-free] (51) to (28);
		\draw [style=fault-free] (50) to (35);
		\draw [style=fault-free] (49) to (36);
		\draw (29.center) to (49);
		\draw (30.center) to (50);
		\draw (31.center) to (51);
		\draw [style=fault-free] (54) to (39);
		\draw [style=fault-free] (51) to (54);
		\draw [style=fault-free] (50) to (52);
		\draw [style=fault-free] (49) to (53);
		\draw [style=fault-free] (53) to (47);
		\draw [style=fault-free] (52) to (46);
		\draw (53) to (43.center);
		\draw (52) to (44.center);
		\draw (54) to (45.center);
		\draw (78) to (55);
		\draw (77) to (62);
		\draw (76) to (63);
		\draw (56.center) to (76);
		\draw (57.center) to (77);
		\draw (58.center) to (78);
		\draw (81) to (66);
		\draw (78) to (81);
		\draw (77) to (79);
		\draw (76) to (80);
		\draw (80) to (74);
		\draw (79) to (73);
		\draw (80) to (70.center);
		\draw (79) to (71.center);
		\draw (81) to (72.center);
		\draw [style=fault-free] (92) to (1);
		\draw [style=fault-free] (92) to (8);
		\draw [style=fault-free] (9) to (91);
		\draw [style=fault-free] (91) to (1);
		\draw [style=fault-free] (9) to (90);
		\draw [style=fault-free] (90) to (8);
		\draw [style=fault-free] (90) to (1);
		\draw [style=fault-free] (98) to (28);
		\draw [style=fault-free] (98) to (35);
		\draw [style=fault-free] (36) to (97);
		\draw [style=fault-free] (97) to (28);
		\draw [style=fault-free] (36) to (96);
		\draw [style=fault-free] (96) to (35);
		\draw [style=fault-free] (96) to (28);
		\draw [style=fault-free] (95) to (55);
		\draw [style=fault-free] (95) to (62);
		\draw [style=fault-free] (63) to (94);
		\draw [style=fault-free] (94) to (55);
		\draw [style=fault-free] (63) to (93);
		\draw [style=fault-free] (93) to (62);
		\draw [style=fault-free] (93) to (55);
		\draw [style=fault-free] (107) to (12);
		\draw [style=fault-free] (20) to (107);
		\draw [style=fault-free] (19) to (106);
		\draw [style=fault-free] (20) to (106);
		\draw [style=fault-free] (20) to (105);
		\draw [style=fault-free] (105) to (12);
		\draw [style=fault-free] (101) to (39);
		\draw [style=fault-free] (47) to (101);
		\draw [style=fault-free] (46) to (100);
		\draw [style=fault-free] (47) to (100);
		\draw [style=fault-free] (47) to (99);
		\draw [style=fault-free] (99) to (39);
		\draw [style=fault-free] (104) to (66);
		\draw [style=fault-free] (74) to (104);
		\draw [style=fault-free] (73) to (103);
		\draw [style=fault-free] (74) to (103);
		\draw [style=fault-free] (74) to (102);
		\draw [style=fault-free] (102) to (66);
	\end{pgfonlayer}
\end{tikzpicture}

\]
However, if we assume that, in the fault-free case, the input state lies in the code space of the code being measured, then the step \emph{is} fault-equivalent. 
Diagrammatically, we can encode this assumption by prepending our equations with an idealised, fault-free encoder. 
We have:
\begin{proposition}
    \[\begin{tikzpicture}
	\begin{pgfonlayer}{nodelayer}
		\node [style=none] (0) at (0, 0) {$\FaultEq$};
		\node [style=X dot] (1) at (6.5, -2) {};
		\node [style=none] (2) at (4.5, 1.5) {};
		\node [style=none] (3) at (4.5, 0.5) {};
		\node [style=none] (4) at (4.5, -1.5) {};
		\node [style=X dot] (8) at (6.5, 0) {};
		\node [style=X dot] (9) at (6.5, 1) {};
		\node [style=none, rotate=90] (10) at (7, -0.5) {...};
		\node [style=none, rotate=90] (11) at (4.75, -0.5) {...};
		\node [style=Z dot] (12) at (9.75, -2) {};
		\node [style=none] (16) at (11.25, 1.5) {};
		\node [style=none] (17) at (11.25, 0.5) {};
		\node [style=none] (18) at (11.25, -1.5) {};
		\node [style=Z dot] (19) at (9.75, 0) {};
		\node [style=Z dot] (20) at (9.75, 1) {};
		\node [style=none, rotate=90] (21) at (10.25, -0.5) {...};
		\node [style=X dot] (22) at (7, 1.5) {};
		\node [style=X dot] (23) at (7, 0.5) {};
		\node [style=X dot] (24) at (7, -1.5) {};
		\node [style=Z dot] (25) at (10.25, 0.5) {};
		\node [style=Z dot] (26) at (10.25, 1.5) {};
		\node [style=Z dot] (27) at (10.25, -1.5) {};
		\node [style=X dot] (28) at (-6.25, -2) {};
		\node [style=none] (29) at (-8.25, 1.5) {};
		\node [style=none] (30) at (-8.25, 0.5) {};
		\node [style=none] (31) at (-8.25, -1.5) {};
		\node [style=X dot] (35) at (-6.25, 0) {};
		\node [style=X dot] (36) at (-6.25, 1) {};
		\node [style=none, rotate=90] (37) at (-5.75, -0.5) {...};
		\node [style=none, rotate=90] (38) at (-8, -0.5) {...};
		\node [style=Z dot] (39) at (-3, -2) {};
		\node [style=none] (43) at (-1.5, 1.5) {};
		\node [style=none] (44) at (-1.5, 0.5) {};
		\node [style=none] (45) at (-1.5, -1.5) {};
		\node [style=Z dot] (46) at (-3, 0) {};
		\node [style=Z dot] (47) at (-3, 1) {};
		\node [style=none, rotate=90] (48) at (-2.5, -0.5) {...};
		\node [style=X dot] (49) at (-5.75, 1.5) {};
		\node [style=X dot] (50) at (-5.75, 0.5) {};
		\node [style=X dot] (51) at (-5.75, -1.5) {};
		\node [style=Z dot] (52) at (-2.5, 0.5) {};
		\node [style=Z dot] (53) at (-2.5, 1.5) {};
		\node [style=Z dot] (54) at (-2.5, -1.5) {};
		\node [style=none] (55) at (4.5, 2) {};
		\node [style=none] (56) at (4.5, -2) {};
		\node [style=none] (57) at (2.5, 0.75) {};
		\node [style=none] (58) at (2.5, -0.75) {};
		\node [style=none] (59) at (2.5, 0.5) {};
		\node [style=none] (60) at (2.5, -0.5) {};
		\node [style=none] (61) at (1.5, 0.5) {};
		\node [style=none] (62) at (1.5, -0.5) {};
		\node [style=none] (63) at (2, 0) {\dots};
		\node [style=none] (64) at (3.5, 0) {$enc$};
		\node [style=none] (65) at (-8.25, 2) {};
		\node [style=none] (66) at (-8.25, -2) {};
		\node [style=none] (67) at (-10.25, 0.75) {};
		\node [style=none] (68) at (-10.25, -0.75) {};
		\node [style=none] (69) at (-10.25, 0.5) {};
		\node [style=none] (70) at (-10.25, -0.5) {};
		\node [style=none] (71) at (-11.25, 0.5) {};
		\node [style=none] (72) at (-11.25, -0.5) {};
		\node [style=none] (73) at (-10.75, 0) {\dots};
		\node [style=none] (74) at (-9.25, 0) {$enc$};
		\node [style=none] (75) at (0, -0.5) {\tiny $d$};
		\node [style=Z dot] (76) at (-7.25, 1) {};
		\node [style=Z dot] (77) at (-7.25, 0) {};
		\node [style=Z dot] (78) at (-7.25, -1) {};
		\node [style=Z dot] (79) at (5.5, 1) {};
		\node [style=Z dot] (80) at (5.5, 0) {};
		\node [style=Z dot] (81) at (5.5, -1) {};
		\node [style=X dot] (82) at (-4.25, 1) {};
		\node [style=X dot] (83) at (-4.25, 0) {};
		\node [style=X dot] (84) at (-4.25, -1) {};
		\node [style=X dot] (85) at (8.5, 1) {};
		\node [style=X dot] (86) at (8.5, 0) {};
		\node [style=X dot] (87) at (8.5, -1) {};
	\end{pgfonlayer}
	\begin{pgfonlayer}{edgelayer}
		\draw (24) to (1);
		\draw (23) to (8);
		\draw (22) to (9);
		\draw (2.center) to (22);
		\draw (3.center) to (23);
		\draw (4.center) to (24);
		\draw (27) to (12);
		\draw (24) to (27);
		\draw (23) to (25);
		\draw (22) to (26);
		\draw (26) to (20);
		\draw (25) to (19);
		\draw (26) to (16.center);
		\draw (25) to (17.center);
		\draw (27) to (18.center);
		\draw [style=fault-free] (51) to (28);
		\draw [style=fault-free] (50) to (35);
		\draw [style=fault-free] (49) to (36);
		\draw (29.center) to (49);
		\draw (30.center) to (50);
		\draw (31.center) to (51);
		\draw [style=fault-free] (54) to (39);
		\draw [style=fault-free] (51) to (54);
		\draw [style=fault-free] (50) to (52);
		\draw [style=fault-free] (49) to (53);
		\draw [style=fault-free] (53) to (47);
		\draw [style=fault-free] (52) to (46);
		\draw (53) to (43.center);
		\draw (52) to (44.center);
		\draw (54) to (45.center);
		\draw [style=fault-free] (55.center)
			 to (56.center)
			 to (58.center)
			 to (57.center)
			 to cycle;
		\draw [style=fault-free] (61.center) to (59.center);
		\draw [style=fault-free] (62.center) to (60.center);
		\draw [style=fault-free] (65.center)
			 to (66.center)
			 to (68.center)
			 to (67.center)
			 to cycle;
		\draw [style=fault-free] (71.center) to (69.center);
		\draw [style=fault-free] (72.center) to (70.center);
		\draw [style=fault-free] (36) to (76);
		\draw [style=fault-free] (76) to (35);
		\draw [style=fault-free] (76) to (28);
		\draw [style=fault-free] (81) to (1);
		\draw [style=fault-free] (81) to (8);
		\draw [style=fault-free] (9) to (80);
		\draw [style=fault-free] (80) to (1);
		\draw [style=fault-free] (9) to (79);
		\draw [style=fault-free] (79) to (8);
		\draw [style=fault-free] (79) to (1);
		\draw [style=fault-free] (78) to (28);
		\draw [style=fault-free] (78) to (35);
		\draw [style=fault-free] (36) to (77);
		\draw [style=fault-free] (77) to (28);
		\draw [style=fault-free] (87) to (12);
		\draw [style=fault-free] (20) to (87);
		\draw [style=fault-free] (19) to (86);
		\draw [style=fault-free] (20) to (86);
		\draw [style=fault-free] (20) to (85);
		\draw [style=fault-free] (85) to (12);
		\draw [style=fault-free] (84) to (39);
		\draw [style=fault-free] (47) to (84);
		\draw [style=fault-free] (46) to (83);
		\draw [style=fault-free] (47) to (83);
		\draw [style=fault-free] (47) to (82);
		\draw [style=fault-free] (82) to (39);
	\end{pgfonlayer}
\end{tikzpicture}
\]
\end{proposition}
\begin{proof}
Fault equivalence from the left-hand side to the right-hand side is immediate, since the right-hand side is obtained by idealising edges and therefore can only reduce fault propagation.

To show fault equivalence from the right-hand side to the left-hand side, we first show that any fault of weight $w$ on the right-hand side induces at most $w$ input faults and at most $w$ output faults.
It suffices to verify this for individual edge faults on the newly non-idealised edges, since faults on distinct edges are independent and their effects compose additively.

We first consider faults on the horizontal edges.
An $X$ fault can be pushed to the input without increasing its weight, and a $Z$ fault can be pushed to the output without increasing its weight.
A $Y$ fault can be decomposed into its $X$ and $Z$ components; pushing these components separately yields one fault on the input and one fault on the output.
Although this results in a fault of total weight two, it induces weight one on the inputs and weight one on the outputs, satisfying the desired bound.

Next, consider faults on the left vertical edges.
An $X$ fault can be pushed to the input without increasing its weight.
A $Y$ or $Z$ fault is equivalent to a $Y$ or $Z$ fault on the input together with a $Z$ fault on the output.
Again, this induces at most one fault on the input and one on the output.
Faults on the right vertical edges satisfy the same property by symmetry.

Thus, every atomic edge fault on the newly non-idealised edges induces at most one input fault and at most one output fault.
By independence, the same bound holds for any combination of such faults.
Consequently, any fault of weight $w$ on the right-hand side induces at most $w_1 \leq w$ input faults and at most $w_2 \leq w$ output faults.

Now let $F$ be a fault of weight $w < d$.
If $w_1 = 0$, then $F$ induces fewer than $d$ output faults and therefore satisfies the fault-equivalence criterion.
If $w_1 > 0$, then $F$ must be detectable, since the left-hand side corresponds to measuring the stabilisers of the code on the code space.
By the distance-$d$ property of the code, any nontrivial fault of weight less than $d$ acting on a code state must be detectable.
Hence, every fault $F$ of weight less than $d$ is either detectable or induces at most $w$ output faults, establishing the fault equivalence.
\end{proof}

This result shows that Steane-style syndrome extraction satisfies a slightly weaker notion of fault-tolerant syndrome extraction than proposed in the paper.
In particular, it cannot be used to measure an arbitrary stabiliser fault-freely in isolation.
To measure a stabiliser $d$-fault-equivalently, one must measure a sufficient collection of stabilisers whose joint measurement defines a code of distance at least $d$.
Consequently, Steane-style syndrome extraction cannot, in general, be used to measure only a strict subset of stabilisers in a single shot at distance $d$.

\subsection{Steane-style syndrome extraction --- Decoding}
\label{appendix:stean-decoding}
As fault-equivalent rewrites preserve distance by \autoref{cor:fe-preserves-dist}, \autoref{prop:optimised-steane} guarantees that the syndrome extraction circuit using the optimised Steane method has distance three. 
Therefore, \autoref{prop:correctability-of-circuits} guarantees that any fault of weight one can be identified and corrected, for example using a lookup table.
To give further insights into the practical application of this guarantee we will discuss the concrete decoder for the optimised Steane-style syndrome extraction for the Steane code.

To understand how syndrome information is calculated in the optimized circuit, it is instructive to first recall the standard procedure for Steane-style syndrome measurements.
In Steane-style syndrome measurements, a 7-bit string $v \in \mathbb{B}^7$ is obtained from the measurements.
This is then multiplied with the parity check matrix:
\[
  H_x = H_z =
  \begin{pmatrix}
    1 & 1 & 1 & 1 & 0 & 0 & 0 \\
    0 & 1 & 1 & 0 & 1 & 1 & 0 \\
    0 & 0 & 1 & 1 & 0 & 1 & 1 \\
  \end{pmatrix}
\]
which gives the exact syndrome information.

For the procedure introduced in \autoref{prop:floq-steane}, we assume that the measurement outcomes are $0$ when they are bent around and treated as inputs.
In this optimized scheme, a 4-bit string $w \in \mathbb{B}^4$ is obtained from the measurement.
As the bent qubits are the $1$-st, $6$-th, and $7$-th, these columns are removed from $H_x$ and $H_z$ leaving us with:
\[
  H_x' = H_z' =
  \begin{pmatrix}
    1 & 1 & 1 & 0 \\
    1 & 1 & 0 & 1 \\
    0 & 1 & 1 & 0 \\
  \end{pmatrix}
\]
Multiplying $w$ with $H_x'$ then gives us the syndrome information.

The procedure described above correctly identifies the syndrome, provided the flag qubit measures $0$ (indicating no detectable fault in the ancilla qubits).
For the original method of Steane this is a reasonable assumption as we can always repeat the preparation of the logical state until success.
However, for the procedure of \autoref{prop:floq-steane}, the circuit is already entangled with the data qubits before the measurement of the flag qubit, making repeat-until-success an unviable option.
Instead, we use a method similar to \textcite{chaoQuantumError2018,reichardtFaulttolerantQuantum2020}, modifying the decoding problem if the flag qubit fires.

First, we note the detecting region that the flag qubit is part of:
\[
  \input{./figures/07-examples/steane/detecting-region-no-data.tikz}
\]
This means that the flag qubit detects any single qubit $X$ error that occurs on the marked edges.
All errors that can be detected by the flag qubit are marked on the following circuit by the corresponding data qubits they propagate to:
\[
  \input{./figures/07-examples/steane/steane-optimised-decoding.tikz}
\]
Faults that propagate to zero or one data qubit are not problematic. 
Neither are the faults at $\bR\{2, 3, 5\}\e$ which propagate to data qubits $X_2X_3X_5$ which is stabiliser equivalent to $X_{6}$.
The main concern here are failures at $\bR\{3, 4\}\e$ and $\bR\{2, 5\}\e$.
Faults at these locations propagate to physical errors $X_{3} X_{4} \sim X_{5} \overline{X}$ and $X_{2} X_{5} \sim X_{1} \overline{X}$, respectively.
They are problematic as the standard correction, unaware of the flag, would misidentify them as an $X_{5}$ and $X_{1}$ error respectively, resulting in a logical error.
Therefore, we have to modify the decoding procedure if the flag qubit fires.
If the flag goes off, the possible resulting errors on the data qubits are $I$, $X_{1} \overline{X}$, $X_2$, $X_4$, $X_{5} \overline{X}$, and $X_6$.
If the next round of syndrome extraction identifies the most likely fault as $I, X_2, X_4, X_6$ or $X_7$, we proceed as usual. 
However, if it identifies the most likely faults to be $X_1$ or $X_5$, we instead apply the correction $X_1 \overline{X}$ or $X_5 \overline{X}$ respectively, given that the flag tells us that these faults are more likely to have occurred.
Assuming that at most a single fault has occurred, this allows for an unambiguous correction.

\end{document}